\documentclass[11pt]{article}
\usepackage{latexsym}
\usepackage{amsfonts,amssymb}
\usepackage{amsmath,amssymb,amsthm,xspace,fullpage}
\usepackage{graphicx}
\usepackage{enumerate}
\usepackage{color}

\newtheorem{theorem}{Theorem}
\newtheorem{lemma}{Lemma}
\newtheorem{definition}{Definition}

\newtheorem{property}{Property}

\usepackage{float}
\usepackage[absolute,overlay]{textpos}
\usepackage[ruled, vlined, linesnumbered]{algorithm2e}
\usepackage{here}

\newcommand{\opt}{\mathsf{OPT}}
\newcommand{\ct}{\mathsf{CT}}
\newcommand{\at}{\mathsf{AT}}
\newcommand{\id}{\mathrm{Id}}

\newcommand{\argmin}{\mathop{\rm arg\,min}\limits}
\newcommand{\hide}[1]{}

\def\squarebox#1{\hbox to #1{\hfill\vbox to #1{\vfill}}}

 \newcommand{\bs}{\bigskip} 
  
 \newcommand{\hs}[1]{\hspace*{ #1 mm}} 

\newcommand{\ignore}[1]{}

\allowdisplaybreaks[1]

\begin{document}
\pagestyle{plain}
\begin{center}
{\Large {\bf 
Almost Linear Time Algorithms for Minsum $k$-Sink Problems on Dynamic Flow Path Networks
}}
\bs\\

{\sc Yuya Higashikawa}$^1$ \hspace{5mm} 
{\sc Naoki Katoh}$^1$ \hspace{5mm} 
{\sc Junichi Teruyama}$^1$ \hspace{5mm} 
{\sc Koji Watase}$^2$ \hspace{5mm} 

\

$^1${School of Social Information Science, University of Hyogo, Japan}; \\
{\tt \{higashikawa,naoki.katoh,junichi.teruyama\}@sis.u-hyogo.ac.jp}\\

$^2${School of Science and Technology, Kwansei Gakuin University, Japan}; \\
{\tt fnt43517@kwansei.ac.jp}

\end{center}
\bs

\begin{abstract}
We address the facility location problems on dynamic flow path networks. 
A \textit{dynamic flow path network} consists of an undirected path with positive edge lengths, 
positive edge capacities, and positive vertex weights. 
A path can be considered as a road, 
an edge length as the distance along the road and 
a vertex weight as the number of people at the site. 
An edge capacity limits the number of people that can enter the edge per unit time. 
In the dynamic flow network, 
given particular points on edges or vertices, called \textit{sinks},
all the people evacuate from the vertices to the sinks as quickly as possible.
The problem is to find the location of
sinks on a dynamic flow path network in such a way that the aggregate evacuation time (i.e., the sum of evacuation times for all the people) to sinks is minimized. 
We consider two models of the problem: the \textit{confluent flow model} and the \textit{non-confluent flow model}.
In the former model, the way of evacuation is restricted so that all the people at a vertex have
to evacuate to the same sink, and in the latter model, there is no such restriction.
In this paper, for both the models,
we develop algorithms which run in almost linear time
regardless of the number of sinks.
It should be stressed that
for the confluent flow model,
our algorithm improves upon the previous result by Benkoczi et al.~[Theoretical Computer Science, 2020], and
one for the non-confluent flow model is the first polynomial time algorithm.
\end{abstract}

\section{Introduction}\label{sec:intro}

Recently, many disasters, such as earthquakes, nuclear plant accidents, volcanic eruptions and flooding,
have struck in many parts of the world,
and it has been recognized that orderly evacuation planning is urgently needed.
A powerful tool for evacuation planning is the {\em dynamic flow model} introduced by Ford and Fulkerson~\cite{ford1958},
which represents movement of commodities over time in a network.
In this model,
we are given a graph 
with \textit{source} vertices and \textit{sink} vertices.
Each source vertex is associated with a positive weight, called a \textit{supply},
each sink vertex is associated with a positive weight, called a \textit{demand}, and each edge is associated with positive length and capacity.
An edge capacity limits the amount of supply that can enter the edge per unit time.
One variant of the dynamic flow problem is the {\em quickest transshipment problem}, of which
the objective is to send exactly the right amount of supply out of sources into sinks with satisfying the demand constraints
in the minimum overall time.
Hoppe and Tardos~\cite{hoppe2000} provided a polynomial time algorithm for this problem in the
case where the transit times are integral.
However, the complexity of their algorithm is very high.
Finding a practical polynomial time solution to this problem is still open.
A reader is referred to a recent survey by Skutella~\cite{skutella2009} on dynamic flows.

This paper discusses a related problem,
called the {\em $k$-sink problem}~\cite{belmonte2015polynomial,BenkocziBHKK18,BenkocziBHKK19,BhattacharyaGHK17,chen2016sink,chen2018minmax,higashikawa2014f,HigashikawaGK15,mamada2006},
of which 
the objective is to find a location of $k$ sinks in a given dynamic flow network so that all the supply is sent to the sinks as quickly as possible.
For the optimality of location, the following two criteria can be naturally considered: the minimization of {\em evacuation completion time} and {\em aggregate evacuation time}
(i.e., {\em average evacuation time}).
We call the $k$-sink problem that requires finding a location of $k$ sinks that minimizes the evacuation completion time (resp. the aggregate evacuation time) the {\em minmax} (resp. {\em minsum}) {\em $k$-sink problem}.
Several papers have studied the minmax $k$-sink problems on dynamic flow networks~\cite{belmonte2015polynomial,BhattacharyaGHK17,chen2016sink,chen2018minmax,higashikawa2014f,HigashikawaGK15,mamada2006}.
On the other hand, the minsum $k$-sink problems on dynamic flow networks have not been studied except for the case of path networks~\cite{BenkocziBHKK18,BenkocziBHKK19,HigashikawaGK15}.

Moreover, there are two models on the way of evacuation.
Under the {\it confluent flow model}, all the supply leaving a vertex must evacuate to the same sink through the same edges,
and under the {\it non-confluent flow model}, there is no such restriction.
To our knowledge, 
all the papers which deal with the $k$-sink problems~\cite{belmonte2015polynomial,BenkocziBHKK18,BenkocziBHKK19,BhattacharyaGHK17,chen2016sink,chen2018minmax,HigashikawaGK15} adopt the confluent flow model.

In order to model the evacuation behavior of people,
it might be natural to treat each supply as a discrete quantity
as in~\cite{hoppe2000,mamada2006}.
Nevertheless, almost all the previous papers on sink problems~\cite{belmonte2015polynomial,BhattacharyaGHK17,chen2016sink,chen2018minmax,higashikawa2014f,HigashikawaGK15} treat 
each supply as a continuous quantity
since 
it is easier for mathematically handling the problems and
the effect is small enough to ignore
when the number of people is large.
Throughout the paper, we also adopt the model with continuous supplies.

In this paper, we study the minsum $k$-sink problems on dynamic flow path networks under both the confluent flow model and the non-confluent flow model. 
A path network can model a coastal area surrounded by the sea and a hilly area, an airplane aisle, a hall way in a building, a street, a highway, etc., to name a few.
For the confluent flow model, the previous best results are
an $O(kn \log^3 n)$ time algorithm for the case with uniform edge capacity in \cite{BenkocziBHKK18}, 
and $O(kn \log^4 n)$ time algorithm for the case with general edge capacities in \cite{BenkocziBHKK19}, 
where $n$ is the number of vertices on path networks.
We develop 
algorithms which run in time 
$\min \{ O(kn \log^2 n),n 2^{O(\sqrt{\log k \log \log n})} \log^2 n \}$ for the case with uniform edge capacity, 
and in time $\min \{ O(kn \log^3 n),n 2^{O(\sqrt{\log k \log \log n})} \log^3 n \}$ for the case with general edge capacities,
respectively.
Thus, our algorithms improve upon the complexities by \cite{BenkocziBHKK18,BenkocziBHKK19} for any value of $k$.
Especially, for the non-confluent flow model, this paper provides the first polynomial time algorithms. 

Since the number of sinks $k$ is at most $n$, 
we confirm $2^{O(\sqrt{\log k \log \log n})} = n^{O( \sqrt{\log \log n/\log n})}= n^{o(1)}$,
which means that
our algorithms are the first ones which run in 
almost linear time 
(i.e., $n^{1+o(1)}$ time) 
regardless of $k$.
The reason why we could achieve almost linear time algorithms for the minsum $k$-sink problems
is that we newly discover a convex property from a novel point of view.
In all the previous papers on the $k$-sink problems, 
the evacuation completion time and the aggregate evacuation time (called $\ct$ and $\at$, respectively) are basically determined as functions in ``distance'': 
Let us consider the case with a 1-sink. 
The values $\ct(x)$ or $\at(x)$ may change as a sink location $x$ moves along edges in the network. 
In contrast, we introduce a new metric for $\ct$ and $\at$ as follows: 
assuming that a sink is fixed and all the supply in the network flows to the sink, for a positive real $z$, 
$\ct(z)$ is the time at which the first $z$ of supply completes its evacuation to the sink and then $\at(z)$ is the integral of $\ct(z)$, i.e., $\at(z)=\int_0^z \ct(t)dt$.
We can observe that $\at(z)$ is convex in $z$
since $\ct(z)$ is increasing in $z$.
Based on the convexity of $\at(z)$, we develop efficient algorithms.

The rest of the paper is organized as follows.
In Section~\ref{sec:preliminary}, we introduce the terms that are used throughout the paper and explain our models.
In Section~\ref{sec:algorithm}, we show that our problem can be reduced to the {\em minimum $k$-link path problem with links satisfying the concave Monge condition}.
This immediately implies by Schieber~\cite{Schieber98} that
the optimal solutions for our problems can be obtained by
solving $\min \{ O(kn),n 2^{O(\sqrt{\log k \log \log n})}\}$ subproblems
of computing the optimal aggregate evacuation time for subpaths, 
in each of which two sinks are located on its endpoints.
Section~\ref{sec:algorithm} subsequently shows 
an overview of the algorithm that solves the above subproblems.
In Section~\ref{sec:data_structure} and~\ref{app:gen_edge_capacity}, 
we introduce novel data structures and then give the algorithm 
which solves each of the above subproblems in $O({\rm poly}\log n)$ time. 
Section~\ref{sec:conclusion} concludes the paper.

\section{Preliminaries}\label{sec:preliminary}

\subsection{Notations}\label{sec:notation}
For two real values $a,b$ with $a < b$, 
let $[a,b]=\{t \in \mathbb{R} \mid a \le t \le b\}$, 
$[a,b)=\{t \in \mathbb{R} \mid a \le t < b\}$, 
$(a,b]=\{t \in \mathbb{R} \mid a < t \le b\}$, and 
$(a,b)=\{t \in \mathbb{R} \mid a < t < b\}$, 
where $\mathbb{R}$ is the set of real values. 
For two integers $i,j$ with $i \leq j$, 
let $[i..j]=\{h \in \mathbb{Z} \mid i \le h \le j\}$, 
where $\mathbb{Z}$ is the set of integers. 
A dynamic flow path network $\mathcal{P}$ is given as a 5-tuple $(P,{\mathbf w},{\mathbf c},{\mathbf l},\tau)$,
where $P$ is a path with vertex set $V = \{v_i \mid i \in [1..n]\}$ and
edge set $E = \{e_i = (v_i, v_{i+1}) \mid i \in [1..n-1]\}$,
${\mathbf w}$ is a vector $\langle w_1, \ldots, w_n \rangle$ of which a component $w_i$ is the {\it weight} of vertex $v_i$ representing the amount of supply (e.g., the number of evacuees, cars) located at $v_i$,
${\mathbf c}$ is a vector $\langle c_1, \ldots, c_{n-1} \rangle$ of which a component $c_i$ is the {\it capacity} of edge $e_i$ representing the upper bound on the flow rate through $e_i$ per unit time, 
${\mathbf l}$ is a vector $\langle \ell_1, \ldots, \ell_{n-1} \rangle$ of which a component $\ell_i$ is the {\it length} of edge $e_i$, and
$\tau$ is the time which unit supply takes to move unit distance on any edge. 

We say a point $p$ lies on path $P = (V, E)$, denoted by $p \in P$, 
if $p$ lies on a vertex $v \in V$ or an edge $e \in E$.
For two points $p, q \in P$, $p \prec q$ means that $p$ lies to the left side of $q$.
For two points $p, q \in P$, $p \preceq q$ means that $p \prec q$ or $p$ and $q$ lie on the same place.
Let us consider two integers $i,j \in [1..n]$ with $i<j$.
We denote by $P_{i,j}$ a {\it subpath} of $P$ from $v_i$ to $v_j$.
Let $L_{i,j}$ be the distance between $v_i$ and $v_j$, i.e., $L_{i,j} = \sum_{h=i}^{j-1}\ell_h$, and
let $C_{i,j}$ be the minimum capacity for all the edges between $v_i$ and $v_j$, i.e., $C_{i,j} = \min\{ c_{h} \mid h \in [i..j-1]\}$.
For $i \in [1..n]$, we denote the sum of weights from $v_1$ to $v_i$ by $W_i = \sum_{j=1}^{i}w_j$.
Note that, given a dynamic flow path network $\mathcal{P}$,
if we construct two lists of $W_i$ and $L_{1,i}$ for all $i \in [1..n]$ in $O(n)$ preprocessing time,
we can obtain $W_i$ for any $i \in [1..n]$ and $L_{i,j} = L_{1,j}-L_{1,i}$ for any $i,j \in [1..n]$ with $i<j$ in $O(1)$ time.
In addition,
$C_{i,j}$ for any $i,j \in [1..n]$ with $i<j$ can be obtained in $O(1)$ time 
with $O(n)$ preprocessing time, which is known as the {\it range minimum query}~\cite{Alstrup2002,Bender2000}.

A {\it $k$-sink} ${\bf x}$ is $k$-tuple $(x_1, \ldots, x_k)$ of points on $P$, where $x_i \prec x_j$ for $i<j$.
We define the function $\id$ for point $p \in P$ as follows:
the value $\id(p)$ is an integer such that $v_{\id(p)} \preceq p \prec v_{\id(p)+1}$ holds.
For a $k$-sink ${\bf x}$ for $\mathcal{P}$, a {\it divider} ${\bf d}$ is $(k-1)$-tuple $(d_1, \ldots, d_{k-1})$ of real values
such that $d_i < d_j$ for $i<j$ and $W_{\id(x_i)} \leq d_i  \leq W_{\id(x_{i+1})}$.
Given a $k$-sink ${\bf x}$ and a divider ${\bf d}$ for $\mathcal{P}$,
the portion $W_{\id(x_i)} - d_{i-1}$ supply that originates from the left side of $x_i$ flows to sink $x_i$, 
and the portion $d_{i} - W_{\id(x_i)}$ supply that originates from the right side of $x_i$ also flows to sink $x_i$.
For instance,
under the non-confluent flow model, 
if $W_{h-1} < d_i < W_{h}$ where $h \in [1..n]$, 
$d_i - W_{h-1}$ of $w_h$ supply at $v_h$ flows to sink $x_{i}$
and the rest of $W_{h} - d_i$ supply to do sink $x_{i+1}$.
The difference between the confluent flow model and the non-confluent flow model is that 
the confluent flow model requires that each value $d_i$ of a divider ${\bf d}$ must take a value in $\{W_1, \ldots, W_n\}$, 
but the non-confluent flow model does not.
For the notation,  we set $d_0 = 0$ and $d_k = W_n$. 

For a dynamic flow path network $\mathcal{P}$, a $k$-sink {\bf x} and a divider {\bf d},
the {\it evacuation completion time} $\ct(\mathcal{P}, {\bf x}, {\bf d})$ is 
the time at which all the supply completes the evacuation.
The {\it aggregate evacuation time} $\at(\mathcal{P}, {\bf x}, {\bf d})$ is that 
the sum of the evacuation completion time for all the supply.
Their explicit definitions are given later.
In this paper, our task is, given a dynamic flow path network $\mathcal{P}$, to find a $k$-sink {\bf x} and a divider {\bf d}
that minimize the aggregate evacuation time $\at(\mathcal{P}, {\bf x}, {\bf d})$
in each evacuation model.

\subsection{Aggregate Evacuation Time on a Path}
For the confluent flow model, 
it is shown in~\cite{BenkocziBHKK19,HigashikawaGK15}
that
for the minsum $k$-sink problems,
there exists an optimal $k$-sink such that all the $k$ sinks are at vertices. 
This fact also holds for the non-confluent flow model. 
Indeed, if a divider {\bf d} is fixed, then we have $k$ subproblems for a 1-sink 
and the optimal sink location for each subproblem is at a vertex. 
Thus, we have the following lemma.
\begin{lemma}[\cite{HigashikawaGK15}]\label{lem:sink_vertex}
For the minsum $k$-sink problem in a dynamic flow path network,
there exists an optimal $k$-sink such that all the $k$ sinks are at vertices
under the confluent/non-confluent flow model.
\end{lemma}
Lemma~\ref{lem:sink_vertex} implies that 
it is enough to consider only the case that every sink is at a vertex. 
Thus, we suppose ${\bf x}=(x_1, \ldots, x_k) \in V^{k}$, where $x_i \prec x_j$ for $i<j$.\\
\noindent
{\bf A simple example with a 1-sink.}
In order to give explicit definitions for the evacuation completion time and the aggregate evacuation time, 
let us consider a simple example for a $1$-sink. 
We are given a dynamic flow path network $\mathcal{P}=(P,{\mathbf w},{\mathbf c},{\mathbf l},\tau)$ 
with $n$ vertices 
and set a unique sink on a vertex $v_i$, that is, ${\bf x} = (v_i)$ and ${\bf d} = ()$ which is the $0$-tuple. 
In this case, all the supply on the left side of $v_i$ (i.e., at $v_1, \ldots, v_{i-1}$) will flow right to sink $v_i$, 
and all the supply on the right side of $v_i$ (i.e., at $v_{i+1}, \ldots, v_n$)  will flow left to sink $v_i$.
Note that in our models all the supply at $v_i$ immediately completes its evacuation at time $0$.

To deal with this case, we introduce some new notations.
Let the function $\theta^{i,+}(z)$ denote the time at which 
the first $z - W_{i}$ of supply on the right side of $v_i$
completes its evacuation to sink $v_i$ (where $\theta^{i,+}(z) = 0$ for $z \in [0,W_i]$).
Higashikawa~\cite{Higashikawa14} 
shows that the value $\theta^{i,+}(W_n)$, 
the evacuation completion time for all the supply on the right side of $v_i$,
is given by the following formula:
\begin{eqnarray}\label{eq:completion_time}
\theta^{i,+}(W_n)
&=& \max\left\{ \frac{W_{n} - W_{j-1} }{ C_{i,j} }+ \tau \cdot L_{i,j} \mid j \in [i+1..n] \right\}.
\end{eqnarray}
Recall that $C_{i,j} = \min \{ c_h \mid h \in [i..j-1] \}$.
We can generalize formula~\eqref{eq:completion_time} to the case with any $z \in [0,W_n]$ as follows:
\begin{eqnarray}\label{eq:function_right}
\theta^{i, +}(z) = \max\{ \theta^{i, +, j}(z) \mid j \in [i+1..n]\}, 
\end{eqnarray}
where $\theta^{i, +, j}(z)$ for $j \in [i+1..n]$ is defined as 
\begin{eqnarray}\label{eq:subfunction_right}
\theta^{i, +, j}(z) = 
\left\{
\begin{array}{ll}
0 & \text{if } z \leq W_{j-1}, \\
\frac{z - W_{j-1}}{C_{i,j}} + \tau\cdot L_{i,j} & \text{if } z > W_{j-1}.
\end{array}
\right.
\end{eqnarray}
Similarly, let $\theta^{i, -}(z)$ denote the time at which 
the first $W_{i-1} - z$ of supply on the left side of $v_i$ 
completes its evacuation to sink $v_i$ (where $\theta^{i,-}(z) = 0$ for $z \in [W_{i-1},W_n]$).
Then, 
\begin{eqnarray}\label{eq:function_left}
\theta^{i, -}(z) = \max\{ \theta^{i, -, j}(z) \mid j \in [1..i-1]\}, 
\end{eqnarray}
where $\theta^{i, -, j}(z)$ is defined as
\begin{eqnarray}\label{eq:subfunction_left}
\theta^{i, -, j}(z) = 
\left\{
\begin{array}{ll}
\frac{W_{j} - z}{C_{j,i}} + \tau\cdot L_{j,i} & \text{if } z < W_j, \\
0 & \text{if } z \geq W_j.
\end{array}
\right.
\end{eqnarray}

The aggregate evacuation times for the supply on the right side 
and the left side of $v_i$ are 
\begin{eqnarray}\label{eq:at_i_sub}
\int_{W_i}^{W_n}\theta^{i, +}(z)dz = \int_{0}^{W_n}\theta^{i, +}(z)dz
\text{ and }
\int_{0}^{W_{i-1}}\theta^{i, -}(z)dz = \int_{0}^{W_n}\theta^{i, -}(z)dz, \notag
\end{eqnarray}
respectively.
Thus, the aggregate evacuation time $\at(\mathcal{P}, (v_i), ())$ is given as 
\begin{eqnarray*}\label{eq:at_i}
\at(\mathcal{P}, (v_i), ()) = \int_{0}^{W_n}\left\{\theta^{i, +}(z) + \theta^{i, -}(z) \right\}dz.
\end{eqnarray*}
\begin{figure}[tb]
  \begin{center}
    \includegraphics[width=7.8cm,pagebox=cropbox,clip]{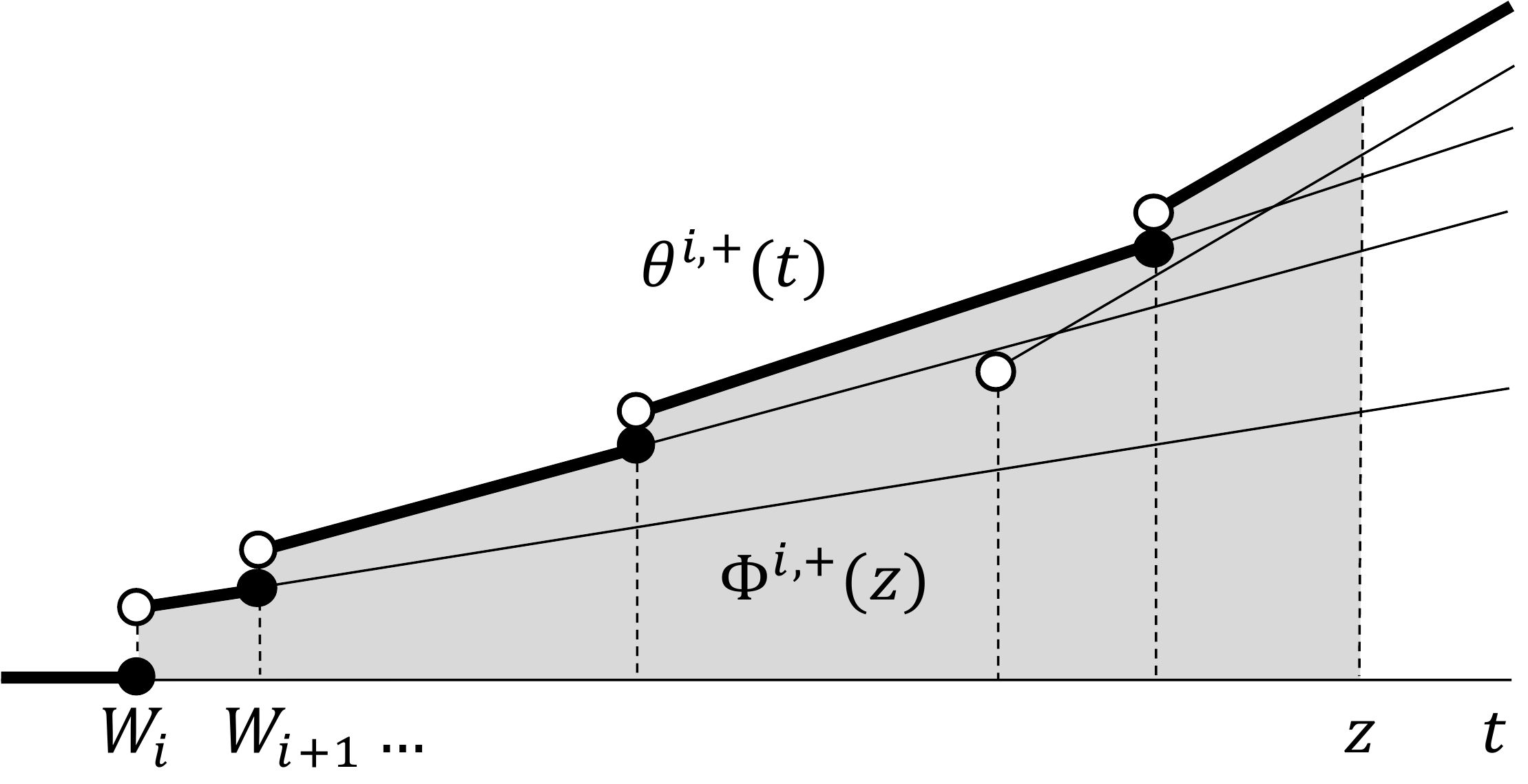}
    \caption{The thick half-open segments indicate function $\theta^{i,+}(t)$ and the gray area indicates $\Phi^{i,+}(z)$ for some $z > W_i$.}
    \label{fig:1}
  \end{center}
\end{figure}
\noindent
{\bf Aggregate evacuation time with a $k$-sink.}
Suppose that we are given a $k$-sink ${\bf x} = (x_1, \ldots, x_k) \in V^{k}$
and a divider ${\bf d} = (d_1, \ldots, d_{k-1})$. 
Recalling the definition of $\id(p)$ for $p \in P$, 
we have $x_i = v_{\id(x_i)}$ for all $i \in [1..k]$. 
In this situation, for each $i \in [1..k]$, 
the first $d_{i} - W_{\id(x_i)}$ of supply on the right side of $x_i$
and
the first $W_{\id(x_i)-1} - d_{i-1}$ of supply on the left side of $x_i$
move to sink $x_i$. 

By the argument of the previous section, 
the aggregate evacuation times for the supply on the right side 
and the left side of $x_i$ are 
\begin{eqnarray}
\int_{W_{\id(x_i)}}^{d_i} \theta^{\id(x_i), +}(z)dz = \int_{0}^{d_i} \theta^{\id(x_i), +}(z)dz 
\text{ and } 
\int_{d_{i-1}}^{W_{\id(x_i)-1}} \theta^{\id(x_{i}),-}(z)dz = \int_{d_{i-1}}^{W_n} \theta^{\id(x_{i}),-}(z)dz, \notag 
\end{eqnarray}
respectively.
In order to give the general form for the above values,
let us denote by $\Phi^{i, +}(z)$ the aggregate evacuation time when
the first $z - W_{i}$ of supply on the right side of $v_i$ flows to sink $v_i$. 
Similarly, we denote by $\Phi^{i, -}(z)$ the aggregate evacuation time when
the first $W_{i-1} - z$ of supply on the left side of $v_i$ flows to sink $v_i$. 
Therefore, we have
\begin{eqnarray}\label{eq:at_gen}
\Phi^{i, +}(z) = \int_{0}^{z} \theta^{i, +}(t)dt \quad \text{and} \quad 
\Phi^{i, -}(z) = \int_{z}^{W_n} \theta^{i,-}(t)dt = \int_{W_n}^{z} -\theta^{i,-}(t)dt
\end{eqnarray}
(see Fig.~\ref{fig:1}).
Let us consider a subpath $P_{\id(x_i), \id(x_{i+1})}$ which is a subpath between sinks $x_{i}$ and $x_{i+1}$. 
The aggregate evacuation time for the supply on $P_{\id(x_i), \id(x_{i+1})}$ is given by 
\[
\int_{0}^{d_i} \theta^{\id(x_i), +}(z)dz + \int_{d_{i}}^{W_n} \theta^{\id(x_{i+1}),-}(z)dz
= \Phi^{\id(x_i), +}(d_i) + \Phi^{\id(x_{i+1}), -}(d_i).
\]
For $i, j \in [1..n]$ with $i < j$, 
let us define 
\begin{eqnarray}\label{eq:at_subpath}
\Phi^{i, j}(z) = \Phi^{i, +}(z) + \Phi^{j,-}(z) = 
\int_{0}^{z} \theta^{i, +}(t)dt + \int_{z}^{W_n} \theta^{j, -}(t)dt
\end{eqnarray}
for $z \in [W_{i}, W_{j-1}]$. 
Then, the aggregate evacuation time 
$\at(\mathcal{P}, {\bf x}, {\bf d})$
is given as 
\begin{eqnarray}\label{eq:at_ksink}
\at(\mathcal{P}, {\bf x}, {\bf d}) = 
\Phi^{\id(x_1),-}(0) + 
\sum_{i=1}^{k-1}\Phi^{\id(x_i), \id(x_{i+1})}(d_i) + 
\Phi^{\id(x_k),+}(W_n).
\end{eqnarray}

In the rest of this section, we show the important properties of $\Phi^{i, j}(z)$.
Let us first confirm that by equation~{(\ref{eq:at_gen})}, both $\Phi^{i, +}(z)$ and $\Phi^{j, -}(z)$ are convex in $z$ since $\theta^{i, +}(z)$ and $-\theta^{j, -}(z)$ are non-decreasing in $z$, 
therefore $\Phi^{i, j}(z)$ is convex in $z$.
On the condition of the minimizer for $\Phi^{i, j}(z)$, we have a more useful lemma.
\begin{lemma}\label{lem:convex}
For any $i,j \in [1..n]$ with $i < j$, 
there uniquely exists
\[
z^* \in \argmin_{z \in [W_i,W_{j-1}]}\max\{\theta^{i,+}(z), \theta^{j,-}(z)\}.
\]
Furthermore,
$\Phi^{i, j}(z)$ is minimized on $[W_i,W_{j-1}]$
when $z=z^*$.
\end{lemma}
\begin{proof}
By equations~\eqref{eq:function_right} and \eqref{eq:subfunction_right},
$\theta^{i,+}(z)$ is strictly increasing in $z \in [W_i,W_n]$.
Similarly, by equations~\eqref{eq:function_left} and \eqref{eq:subfunction_left},
$\theta^{j,-}(z)$ is strictly decreasing in $z \in [0,W_{j-1}]$.
Thus there uniquely exists $z^* \in [W_i,W_{j-1}]$.

We then see that for any $z' \in [W_{i}, z^*]$,
\begin{eqnarray*}
\Phi^{i, j}(z^*) - \Phi^{i, j}(z') &=& \Phi^{i,+}(z^*) + \Phi^{j,-}(z^*) - (\Phi^{i,+}(z') + \Phi^{j,-}(z')) \\
&=& \int_{z'}^{z^*} \theta^{i,+}(t)dt - \int_{z'}^{z^*} \theta^{j,-}(t)dt \\
&=& \int_{z'}^{z^*} \left\{ \theta^{i,+}(t) - \theta^{j,-}(t) \right\} dt \le 0,
\end{eqnarray*}
and for any $z'' \in [z^*, W_{j-1}]$,
\begin{eqnarray*}
\Phi^{i, j}(z^*) - \Phi^{i, j}(z'') &=& \Phi^{i,+}(z^*) + \Phi^{j,-}(z^*) - (\Phi^{i,+}(z'') + \Phi^{j,-}(z')) \\
&=& \int_{z''}^{z^*} \theta^{i,+}(t)dt - \int_{z''}^{z^*} \theta^{j,-}(t)dt \\
&=& -\int_{z^*}^{z''} \left\{ \theta^{i,+}(t) - \theta^{j,-}(t) \right\} dt \le 0,
\end{eqnarray*}
which implies that $z^*$ minimizes $\Phi^{i, j}(z)$ on $[W_{i}, W_{j-1}]$.
\end{proof}

In the following sections,
such $z^*$ is called the {\it pseudo-intersection point}\footnote{The reason why we adopt a term ``pseudo-intersection'' is that two functions $\theta^{i,+}(z)$ and $\theta^{j,-}(z)$ are not continuous in general while ``intersection'' is usually defined for continuous functions.} 
of $\theta^{i,+}(z)$ and $\theta^{j,-}(z)$, and
we say that $\theta^{i,+}(z)$ and $\theta^{j,-}(z)$ {\it pseudo-intersect} 
on $[W_i,W_{j-1}]$ at $z^*$.


\section{Algorithms}\label{sec:algorithm}

In order to solve our problems, we reduce them to {\it minimum $k$-link path problems}. 
In the minimum $k$-link path problems, we are given a weighted complete directed acyclic graph (DAG) $G = (V', E', w')$ 
with $V' = \{v'_{i} \mid i \in [1..n]\}$ and $E' = \{(v'_i, v'_j) \mid i,j \in [1..n], i < j\}$. Each edge $(v'_i, v'_j)$ is associated with weight $w'(i, j)$. 
We call a path in $G$ a {\it $k$-link} path if the path contains exactly $k$ edges. 
The task is to find a $k$-link path $(v'_{a_0} = v'_1, v'_{a_1}, v'_{a_2}, \ldots, v'_{a_{k-1}}, v'_{a_k} = v'_n)$ 
from $v'_1$ to $v'_n$
that minimizes the sum of weights of $k$ edges, $\sum_{i=1}^{k}w'(a_{i-1}, a_i)$.
If the weight function $w'$ satisfies the {\it concave Monge property},
then we can solve the minimum $k$-link path problems
in almost linear time regardless of $k$.
\begin{definition}[Concave Monge property]
We say function $f: {\mathbb Z} \times {\mathbb Z} \rightarrow {\mathbb R}$ 
satisfies the concave Monge property 
if for any integers $i,j$ with $i+1<j$, $f(i, j) + f(i+1, j+1) \leq f(i+1, j) + f(i, j+1)$ holds.
\end{definition}
\begin{lemma}[\cite{Schieber98}]\label{lem:klink}
Given a weighted complete DAG with $n$ vertices, if the weight function satisfies the concave Monge property, 
then there exists an algorithm that solves the minimum $k$-link path problem 
in time $\min\{O(kn),n 2^{O(\sqrt{\log k \log\log n})}\}$.
\end{lemma}


We describe how to reduce 
the $k$-sink problem on
a dynamic flow path network $\mathcal{P}=(P=(V, E),{\mathbf w},{\mathbf c},{\mathbf l},\tau)$
with $n$ vertices 
to the minimum $(k+1)$-link path problem on a weighted complete DAG $G = (V', E', w')$.
We prepare a weighted complete DAG $G = (V', E', w')$ with $n+2$ vertices,
where $V' = \{v'_i \mid i \in [0..n+1]\}$ and $E' = \{ (v'_i, v'_j) \mid i,j \in [0..n+1], i<j\}$.
We set the weight function $w'$ as 
\begin{eqnarray}\label{eq:klink_weight}
w'(i,j) = \left\{
\begin{array}{lll}
\opt(i, j)& \quad & i, j \in [1..n], i < j, \\
\Phi^{i, +}(W_n) & &i \in [1..n] \text{ and } j = n+1, \\
\Phi^{j, -}(0) & &i = 0 \text{ and } j \in [1..n], \\
\infty & & i = 0 \text{ and } j=n+1,
\end{array}
\right.
\end{eqnarray}
where 
$\opt(i, j) = \min_{z \in [W_{i}, W_{j-1}]} \Phi^{i, j}(z)$. 

Now, on a weighted complete DAG $G$ made as above,
let us consider a $(k+1)$-link path $(v'_{a_0}=v'_0, v'_{a_1}, \ldots, v'_{a_{k}}, v'_{a_{k+1}}=v'_{n+1})$ from $v'_0$ to $v'_{n+1}$,
where $a_1, \ldots, a_k$ are integers satisfying $0 < a_1 < a_2 < \cdots < a_k < n+1$.
The sum of weights of this $(k+1)$-link path is
\begin{eqnarray*}
\sum_{i=0}^{k}w'(a_{i}, a_{i+1}) 
&= &
\Phi^{a_1, -}(0) + \sum_{i=1}^{k-1} \opt(a_{i}, a_{i+1}) + \Phi^{a_{k}, +}(W_n).
\end{eqnarray*}
This value is equivalent to $\min_{{\bf d}} \at({\cal P}, {\bf x}, {\bf d})$ for a $k$-sink ${\bf x} = (v_{a_1}, v_{a_2}, \ldots, v_{a_{k}})$
(recall equation~\eqref{eq:at_ksink}),
which implies that a minimum $(k+1)$-link path on $G$ 
corresponds to an optimal $k$-sink location for a dynamic flow path network $\mathcal{P}$.


We show in the following lemma that the function $w'$ defined as formula~{(\ref{eq:klink_weight})}
satisfies the concave Monge property under both of evacuation models.
\begin{lemma}\label{lem:concave_monge}
The weight function $w'$ defined as formula~{(\ref{eq:klink_weight})}
satisfies the concave Monge property 
under the confluent/non-confluent flow model.
\end{lemma}
\begin{proof}

If we show that, for any $i,j \in [0..n]$ with $i<j$, 
\begin{equation}\label{eq:concave_monge}
 w'(i, j) + w'(i+1, j+1) \leq w'(i, j+1) + w'(i+1, j) 
\end{equation}
holds, then the proof completes. 
Note that the condition~\eqref{eq:concave_monge} holds for $i=0$ and $j=n$, 
because the right-hand side of~\eqref{eq:concave_monge} contains ${w'(0,n+1) = \infty}$ 
and other terms are finite. 

\noindent
{\bf Proof for the non-confluent flow model:} 
First, let us consider case of $0 < i < j < n$. 
By formula~\eqref{eq:klink_weight}, for any $(i', j') \in \{(i,j), (i, j+1), (i+1,j),(i+1, j+1)\}$, 
we have $w'(i',j') = \opt(i', j')$. 
Under the non-confluent flow model, for any $i,j \in [1..n]$ with $i<j$, 
$\opt(i,j) = \min_{z \in [W_{i}, W_{j-1}]} \Phi^{i,j}(z)$. 
Lemma~\ref{lem:convex} implies that $\Phi^{i,j}(z)$ is minimized when $z$ is the pseudo-intersection point of $\theta^{i,+}(z)$ and $\theta^{j,-}(z)$. 
For any $i,j \in [1..n]$ with $i<j$, let $\alpha^{i,j}$ denote 
the pseudo-intersection point of $\theta^{i,+}(z)$ and $\theta^{j,-}(z)$. 
Thus, we have 
\begin{eqnarray}\label{eq:opt} 
w'(i,j) = \opt(i, j) = \Phi^{i,j}(\alpha^{i,j}) = 
\int_{0}^{\alpha^{i,j}} \theta^{i, +}(z)dz + \int_{\alpha^{i,j}}^{W_n}\theta^{j, -}(z)dz.
\end{eqnarray}

We give two lemmas in order to show the concave Monge condition.
\begin{lemma}\label{lem:monotone}
For any integer $i \in [1..n-1]$ and any $z \in [0, W_n]$, 
\begin{eqnarray*}
\theta^{i,+}(z) \geq \theta^{i+1,+}(z) \text{ and }
\theta^{i,-}(z) \leq \theta^{i+1,-}(z)
\end{eqnarray*}
hold.
\end{lemma}
\begin{proof}
We give the proof only of $\theta^{i,+}(z) \geq \theta^{i+1,+}(z)$ 
because the other case can be shown in a similar way. 
By formula~(\ref{eq:subfunction_right}), 
for any $j \in [i+2..n]$, we have
\begin{eqnarray*}
\theta^{i, +, j}(z) - \theta^{i+1, +, j}(z) =
\left\{
\begin{array}{ll}
0 & \text{if } z \leq W_{j-1}, \\
\frac{z - W_{j-1}}{C_{i,j}} + \tau\cdot L_{i,j} & \text{if } W_{j-1} < z \leq W_{j}, \\
\frac{(z - W_{j-1}) (C_{i+1,j} - C_{i,j})}{C_{i,j}C_{i+1,j}}+ \tau \cdot \ell_{i}
& \text{if } z > W_{j}.
\end{array}
\right.
\end{eqnarray*}
Since $C_{i+1,j} - C_{i,j} =
\min \{ c_h \mid h \in [i+1..j-1] \} - \min \{ c_h \mid h \in [i..j-1] \}
\geq 0$, $\theta^{i, +, j}(z) - \theta^{i+1, +,j}(z) \geq 0$ holds.
Therefore, we have $\theta^{i,+}(z) \geq \theta^{i+1,+}(z)$ since 
$\theta^{i, +}(z) = \max\{ \theta^{i, +, j}(z) \mid j \in [i+1..n]\}$
by formula~(\ref{eq:function_right}). 
\end{proof}

\begin{lemma}\label{lem:cross}
For any $i,j \in [1..n]$ with $i<j$, 
\begin{eqnarray*}
\alpha^{i,j} \leq \alpha^{i+1, j} \leq \alpha^{i+1, j+1} \text{ and }
\alpha^{i,j} \leq \alpha^{i, j+1} \leq \alpha^{i+1, j+1}
\end{eqnarray*}
hold.
\end{lemma}
\begin{proof}
We give the proof only of $\alpha^{i,j} \leq \alpha^{i+1, j}$ 
because other cases can be shown in a similar way. 
For any $i,j \in [1..n]$ with $i<j$ and positive constant $\epsilon$, 
we have
\[
\theta^{i+1,+}(\alpha^{i,j} - \epsilon) \leq \theta^{i,+}(\alpha^{i,j} - \epsilon) < \theta^{j,-}(\alpha^{i,j} - \epsilon)
\]
because $\theta^{i,+}(z) \geq \theta^{i+1,+}(z)$ 
holds by Lemma~\ref{lem:monotone} and $\theta^{j,-}(z)$ is a non-increasing function. 
It implies that $\alpha^{i,j} \leq \alpha^{i+1, j}$ holds and the proof completes.
\end{proof}

For any $i,j \in [1..n-1]$ with $i<j$, equation~(\ref{eq:opt}) and Lemma~\ref{lem:cross} lead that 
\begin{eqnarray}\label{eq:delta_non_confluent}
& & 
w'(i, j+1) + w'(i+1, j) - w'(i, j) - w'(i+1, j+1) \nonumber \\ 
&=& \Phi^{i,j+1}(\alpha^{i,j+1}) + \Phi^{i+1,j}(\alpha^{i+1,j})
- \Phi^{i,j}(\alpha^{i,j}) - \Phi^{i+1,j+1}(\alpha^{i+1,j+1}) \nonumber \\ 
&=& \int_{\alpha^{i,j}}^{\alpha^{i, j+1}} \theta^{i,+}(z)dz
+ \int_{\alpha^{i, j+1}}^{\alpha^{i+1, j+1}}\theta^{j+1,-}(z)dz \nonumber \\
&&\quad\quad\quad\quad - \int_{\alpha^{i, j}}^{\alpha^{i+1, j}}\theta^{j,-}(z)dz
- \int_{\alpha^{i+1, j}}^{\alpha^{i+1, j+1}} \theta^{i+1,+}(z)dz. 
\end{eqnarray}

Now, we show that for any $z \in [\alpha^{i,j}, \alpha^{i+1,j+1})$, 
\[
\min\{\theta^{i,+}(z), \theta^{j+1,-}(z) \} \geq \max\{ \theta^{j,-}(z), \theta^{i+1,+}(z) \}
\]
holds. 
First, for any $z \in [0, W_n]$, we have 
$\theta^{i,+}(z) \geq \theta^{i+1,+}(z)$ and $\theta^{j,-}(z) \leq \theta^{j+1,-}(z)$ hold 
by Lemma~\ref{lem:monotone}.
For any $z \geq \alpha^{i,j}$, we have $\theta^{i,+}(z) \geq \theta^{j,-}(z)$ holds 
since $\alpha^{i,j}$ is the pseudo-intersection point of $\theta^{i,+}(z)$ and $\theta^{j,-}(z)$. 
Similarly, for any $z < \alpha^{i+1,j+1}$,  we have $\theta^{i+1,+}(z) \leq \theta^{j+1,-}(z)$. 
Therefore, any $z \in [\alpha^{i,j}, \alpha^{i+1,j+1})$, 
$\min\{\theta^{i,+}(z), \theta^{j+1,-}(z) \} \geq \max\{ \theta^{j,-}(z), \theta^{i+1,+}(z) \}$ holds.

Thus, equation~{(\ref{eq:delta_non_confluent})} continues as 
\begin{eqnarray*}
& & 
w'(i, j+1) + w'(i+1, j) - w'(i, j) - w'(i+1, j+1) \nonumber \\ 
&\geq&  \int_{\alpha^{i, j}}^{\alpha^{i+1, j+1}} \min\{ \theta^{i,+}(z), \theta^{j+1,-}(z) \} - \int_{\alpha^{i, j}}^{\alpha^{i+1, j+1}} \max\{ \theta^{j,-}(z), \theta^{i+1,+}(z) \} dz \\
&\geq& 0, 
\end{eqnarray*}
and then condition~\eqref{eq:concave_monge} holds for any $i, j$ with $0 < i < j < n$. 

Next, let us consider the case of $i = 0$ and $j \in [1..n-1]$. 
Recall that $w'(0, j) = \Phi^{j,-}(0)$ and $w'(0, j+1) = \Phi^{j+1,-}(0)$
by formula~\eqref{eq:klink_weight}.
In this case, we have
\begin{eqnarray*}
& & w'(0, j+1) + w'(1, j) - w'(0, j) - w'(1, j+1) \nonumber \\ 
&=& \Phi^{j+1,-}(0) + \Phi^{1,j}(\alpha^{1,j}) - \Phi^{j,-}(0) - \Phi^{1,j+1}(\alpha^{1,j+1})\\
&=& \int_{0}^{W_n}\theta^{j+1,-}(z)dz + 
\int_{0}^{\alpha^{1,j}} \theta^{1, +}(z)dz + \int_{\alpha^{1,j}}^{W_n}\theta^{j, -}(z)dz \\
& &\quad \quad - \int_{0}^{W_n}\theta^{j,-}(z)dz - 
\int_{0}^{\alpha^{1,j+1}} \theta^{1, +}(z)dz - \int_{\alpha^{1,j+1}}^{W_n}\theta^{j+1, -}(z)dz \\
&=& \int_{0}^{\alpha^{1,j+1}}\theta^{j+1,-}(z)dz - \int_{0}^{\alpha^{1,j}}\theta^{j, -}(z)dz - \int_{\alpha^{1,j}}^{\alpha^{1,j+1}} \theta^{1, +}(z)dz, 
\end{eqnarray*}
where the last equality uses $\alpha^{1,j} \leq \alpha^{1,j+1}$ by Lemma~\ref{lem:cross}.
By Lemma~\ref{lem:monotone}, we have $\theta^{j+1,-}(z) \geq \theta^{j,-}(z)$ for any $z \in [0, W_n]$.
Using the same argument for the previous case, for any $z < \alpha^{1,j+1}$, we have $\theta^{1,+}(z) < \theta^{j+1,-}(z)$. 
Thus, we have
\begin{eqnarray*}
& & 
w'(0, j+1) + w'(1, j) - w'(0, j) - w'(1, j+1) \nonumber \\ 
&=& \int_{0}^{\alpha^{1,j+1}}\theta^{j+1,-}(z)dz - \int_{0}^{\alpha^{1,j}}\theta^{j, -}(z)dz - \int_{\alpha^{1,j}}^{\alpha^{1,j+1}} \theta^{1, +}(z)dz \\
&=& 
\int_{0}^{\alpha^{1,j}}\left\{ \theta^{j+1,-}(z) - \theta^{j, -}(z) \right\}dz +  \int_{\alpha^{1,j}}^{\alpha^{1,j+1}} \left\{ \theta^{j+1,-}(z) - \theta^{1, +}(z)\right\}dz \geq 0.
\end{eqnarray*}

The rest of the proof is the case of $j = n$ and $i \in [1..n-1]$. 
Recall that $w'(i, n+1) = \Phi^{i,+}(W_n)$ and $w'(i+1, n+1) = \Phi^{i+1,+}(W_n)$
by formula~\eqref{eq:klink_weight}.
Similar to the second case, 
we use facts that $\alpha^{i,n} \leq \alpha^{i+1,n}$ by Lemma~\ref{lem:cross}, 
$\theta^{i,+}(z) \geq \theta^{i+1,+}(z)$ for any $z \in [0, W_n]$ by Lemma~\ref{lem:monotone} and 
$\theta^{i,+}(z) \geq \theta^{n,-}(z)$ for any $z \geq \alpha^{i,n}$. 
Then, we have 
\begin{eqnarray*}
& & w'(i, n+1) + w'(i+1, n) - w'(i, n) - w'(i, n+1) \\ 
&=& \Phi^{i,+}(W_n) + \Phi^{i+1,n}(\alpha^{i+1,n}) - \Phi^{i,n}(\alpha^{i,n}) - \Phi^{i+1,+}(W_n)\\
&=& \int_{0}^{W_n}\theta^{i,+}(z)dz + 
\int_{0}^{\alpha^{i+1,n}} \theta^{i+1, +}(z)dz + \int_{\alpha^{i+1,n}}^{W_n}\theta^{n,-}(z)dz \\
& &\quad \quad - \int_{0}^{\alpha^{i,n}}\theta^{i,+}(z)dz - 
\int_{\alpha^{i,n}}^{W_n} \theta^{n, -}(z)dz - \int_{0}^{W_n}\theta^{i+1,+}(z)dz \\
&=& \int_{\alpha^{i,n}}^{W_n}\theta^{i,+}(z)dz - \int_{\alpha^{i,n}}^{\alpha^{i+1,n}} \theta^{n, -}(z)dz - \int_{\alpha^{i+1,n}}^{W_n}\theta^{i+1,+}(z)dz \\
&=& 
\int_{\alpha^{i,n}}^{\alpha^{i+1,n}} \left\{ \theta^{i,+}(z) - \theta^{n,-}(z) \right\} dz 
+ \int_{\alpha^{i+1,n}}^{W_n}\left\{ \theta^{i,+}(z) - \theta^{i+1,+}(z)\right\}dz \geq 0.
\end{eqnarray*}

Thus, for any $i, j \in [0..n]$ with $i<j$,  condition~\eqref{eq:concave_monge} holds. 
It implies that the function $w'$ satisfies the concave Monge condition.

\noindent
{\bf Proof for the confluent flow model:}
Similar to the case for the non-comfluent flow model,
first, let us consider case of $0 < i < j < n$. 
In this case, for any $(i', j') \in \{(i,j), (i, j+1), (i+1,j),(i+1, j+1)\}$, 
we have $w'(i',j') = \opt(i', j')$. 
Under the confluent flow model, each element of a divider ${\bf d}$ 
should take one of $n$ values $W_1, W_2, \ldots, W_n$. 
This implies that $\opt(i,j)$ is one of 
$\Phi^{i,j}(z)$ for $z = W_i, W_{i+1}, \ldots, W_{j}$. 
For any $i,j \in [1..n]$ with $i < j$, let $\beta^{i,j}$ denote a real value $z$ 
that minimizes $\Phi^{i,j}(z)$ for $z \in \{ W_h \mid h \in [i..j] \}$. 
Thus, we have
\begin{eqnarray}\label{eq:opt_conf}
\opt(i, j) = \Phi^{i, j}(\beta^{i, j}) = 
\int_{0}^{\beta^{i, j}} \theta^{i,+}(z)dz + \int_{\beta^{i, j}}^{W_n}\theta^{j,-}(z)dz
\end{eqnarray}
under the confluent flow model. 
Since $\Phi^{i,j}(z)$ is convex by Lemma~\ref{lem:convex}, 
$\beta^{i, j}$ is one of the following two values: 
(i) the largest value $z \in \{ W_h \mid h \in [i..j] \}$ with smaller than $\alpha^{i,j}$ or
(ii) the smallest value $z \in \{ W_h \mid h \in [i..j] \}$ with larger than or equal to $\alpha^{i,j}$.
This fact and Lemma~\ref{lem:cross} imply the following lemma. 
\begin{lemma}\label{lem:cross_confluent} 
For any $i,j \in [1..n]$ with $i<j$, 
\begin{eqnarray*}
\beta^{i,j} \leq \beta^{i+1, j}\leq \beta^{i+1, j+1} \text{ and }
\beta^{i,j} \leq \beta^{i, j+1} \leq \beta^{i+1, j+1} 
\end{eqnarray*} 
hold. 
\end{lemma} 

Let us suppose that $\beta^{i+1, j} \leq \beta^{i, j+1}$ holds. 
(Note that one can prove for the case of $\beta^{i+1, j} > \beta^{i, j+1}$ in a similar way.) 
Lemma~\ref{lem:cross_confluent} implies that
\begin{eqnarray}
& & 
w'(i, j+1) + w'(i+1, j) - w'(i, j) - w'(i+1, j+1) \nonumber \\ 
& = &
\int_{\beta^{i, j}}^{\beta^{i, j+1}} \theta^{i,+}(z)dz
+ \int_{\beta^{i, j+1}}^{\beta^{i+1, j+1}}\theta^{j+1,-}(z)dz
- \int_{\beta^{i, j}}^{\beta^{i+1, j}}\theta^{j,-}(z)dz
- \int_{\beta^{i+1, j}}^{\beta^{i+1, j+1}} \theta^{i+1,+}(z)dz \nonumber \\
& = &
\int_{\beta^{i, j}}^{\beta^{i+1, j}} \left\{ \theta^{i,+}(z) - \theta^{j,-}(z) \right\}dz
+ \int_{\beta^{i+1, j}}^{\beta^{i, j+1}} \left\{ \theta^{i,+}(z) - \theta^{i+1,+}(z) \right\}dz \nonumber \\
&& \quad\quad\quad\quad + 
\int_{\beta^{i, j+1}}^{\beta{i+1,j+1}} \left\{ \theta^{j+1,-}(z) - \theta^{i+1,+}(z) \right\}dz.
\end{eqnarray}
The first term
$\int_{\beta^{i, j}}^{\beta^{i+1, j}} \left\{ \theta^{i,+}(z) - \theta^{j,-}(z) \right\}dz$
is non-negative. 
Indeed, by the optimality of $\Phi^{i,j}(z)$, we have
\begin{eqnarray*}
\Phi^{i, j}(\beta^{i+1, j}) & \geq & \opt(i, j) = \Phi^{i, j}(\beta^{i, j}) \\
\int_{0}^{\beta_{i+1, j}} \theta^{i,+}(z)dz
+ \int_{\beta_{i+1, j}}^{W_n}\theta^{j,-}(z)dz
& \geq &
\int_{0}^{\beta^{i, j}} \theta^{i,+}(z)dz + \int_{\beta^{i, j}}^{W_n}\theta^{j,-}(z)dz \\
\int_{\beta^{i, j}}^{\beta^{i+1, j}} \left\{ \theta^{i,+}(z) - \theta^{j,-}(z) \right\}dz &\geq& 0.
\end{eqnarray*}
Similarly, the third term $\int_{\beta^{i, j+1}}^{\beta{i+1,j+1}} \left\{ \theta^{j+1,-}(z) - \theta^{i+1,-}(z) \right\}dz$
is non-negative since we have 
$\Phi^{i+1, j+1}(\beta^{i, j+1}) \geq \opt(i+1, j+1) = \Phi^{i+1, j}(\beta^{i+1, j+1})$. 
The second term $\int_{\beta{i+1,j}}^{\beta{i,j+1}} \left\{ \theta^{i,+}(z) - \theta^{i+1,+}(z) \right\}dz$ is 
also non-negative
because for any $z \in [0, W_n]$, $\theta^{i,+}(z) \geq \theta^{i+1,+}(z)$ holds by Lemma~{\ref{lem:monotone}}.
Therefore, we have 
\[
w'(i, j+1) + w'(i+1, j) - w'(i, j) - w'(i+1, j+1) \geq 0,
\]
that is, condition~\eqref{eq:concave_monge} holds.

Next, let us consider the case of $i = 0$ and $j \in [1..n]$. 
In this case, we have
\begin{eqnarray*}
& & w'(0, j+1) + w'(1, j) - w'(0, j) - w'(1, j+1) \\
&=& \Phi^{j+1,-}(0) + \Phi^{1,j}(\beta^{1,j}) - \Phi^{j,-}(0) - \Phi^{1,j+1}(\beta^{1,j+1})\\
&=& \int_{0}^{W_n}\theta^{j+1,-}(z)dz + 
\int_{0}^{\beta^{1,j}} \theta^{1, +}(z)dz + \int_{\beta^{1,j}}^{W_n}\theta^{j, -}(z)dz 
- \int_{0}^{W_n}\theta^{j,-}(z)dz 
- \Phi^{1,j+1}(\beta^{1,j+1}) \\
&=& \int_{0}^{W_n}\theta^{j+1,-}(z)dz + 
\int_{0}^{\beta^{1,j}} \theta^{1, +}(z)dz - \int_{0}^{\beta^{1,j}}\theta^{j,-}(z)dz 
- \Phi^{1,j+1}(\beta^{1,j+1}) \\
&=& 
\int_{0}^{\beta^{1,j}}\left\{ \theta^{j+1,-}(z) - \theta^{j, -}(z) \right\}dz
+ \Phi^{1, j+1}(\beta^{1,j}) - \Phi^{1, j+1}(\beta^{1,j+1}), \end{eqnarray*}
where the last equality uses 
$\Phi^{1, j+1}(\beta^{1,j}) = 
\int_{0}^{\beta^{1,j}} \theta^{1, +}(z)dz - 
\int_{\beta^{1,j}}^{W_n}\theta^{j+1,-}(z)dz$.
The first term $\int_{0}^{\beta^{1,j}}\left\{ \theta^{j+1,-}(z) - \theta^{j, -}(z) \right\}dz$ is non-negative
because $\theta^{j+1,-}(z) \geq \theta^{j,-}(z)$ holds by Lemma~\ref{lem:monotone}. 
Since the function $\Phi^{1,j+1}(z)$ is minimized when $z = \beta^{1,j+1}$, we have
$\Phi^{1, j+1}(\beta^{1,j}) - \Phi^{1, j+1}(\beta^{1,j+1}) \geq 0$.
Thus, 
\[
w'(0, j+1) + w'(1, j) - w'(0, j) - w'(1, j+1) \geq 0
\]
holds. 

The rest of the proof is the case of $j = n$ and $i \in [0..n-1]$. 
Similar to the previous case, we have
\begin{eqnarray*}
& &w'(i, n+1) + w'(i+1, n) - w'(i, n) - w'(i+1, n+1) \\
&=& \Phi^{i,+}(W_n) + \Phi^{i+1,n}(\beta^{i+1,n}) - \Phi^{i,n}(\beta^{i,n}) - \Phi^{i+1,+}(W_n)\\
&=& \int_{0}^{W_n}\theta^{i,+}(z)dz + 
\int_{0}^{\beta^{i+1,n}} \theta^{i+1, +}(z)dz + \int_{\beta^{i+1,n}}^{W_n}\theta^{n,-}(z)dz
- \Phi^{i,n}(\beta^{i,n}) - \int_{0}^{W_n}\theta^{i+1,+}(z)dz \\
&=& \int_{0}^{W_n}\theta^{i,+}(z)dz 
- \int_{\beta^{i+1,n}}^{W_n}\theta^{i+1,+}(z)dz
+ \int_{\beta^{i+1,n}}^{W_n}\theta^{n,-}(z)dz 
- \Phi^{i,n}(\beta^{i,n}) \\
&=& 
\int_{\beta^{i+1,n}}^{W_n}\left\{ \theta^{i,+}(z) - \theta^{i+1,+}(z)\right\}dz + 
\Phi^{i,n}(\beta^{i+1,n}) - \Phi^{i,n}(\beta^{i,n}), 
\end{eqnarray*}
where the last equality uses 
$\Phi^{i, n}(\beta^{i+1,n}) = 
\int_{0}^{\beta^{i+1,n}} \theta^{i+1, +}(z)dz - 
\int_{\beta^{i+1,n}}^{W_n}\theta^{n,-}(z)dz$. 
We also have $\theta^{i,+}(z) \geq \theta^{i+1,+}(z)$ and then the first term is non-negative. 
Since the function $\Phi^{i,n}(z)$ is minimized when $z = \beta^{i,n}$, we have
$\Phi^{i,n}(\beta^{i+1,n}) - \Phi^{i,n}(\beta^{i+1,n}) \geq 0$. 
Thus, $w'(i, n+1) + w'(i+1, n) - w'(i, n) - w'(i+1, n+1) \geq 0$ holds.

Thus, for any $i, j \in [0..n]$ with $i<j$, condition~\eqref{eq:concave_monge} holds. 
It implies that the function $w'$ satisfies the concave Monge condition. 
\end{proof}

Lemmas~\ref{lem:klink} and~\ref{lem:concave_monge} imply that 
if we can evaluate $w'(i,j)$ 
in time at most $t$ for any $i,j \in [0..n+1]$ with $i < j$, 
then we can solve the $k$-sink problem in time $\min\{O(knt), n2^{O(\sqrt{\log k \log\log n})} t \}$.

In order to obtain $w'(i,j)$ for any $i,j \in [0..n+1]$ with $i < j$ in $O({\rm poly}\log n)$ time,
we introduce novel data structures and some modules using them. 
Basically, we construct a {\it segment tree}~\cite{deBerg2010} ${\cal T}$ with root $\rho$
such that its leaves correspond to indices of vertices of $P$ arranged from left to right and its height is $O(\log n)$.
For a node $u \in {\cal T}$, let ${\cal T}_u$ denote the subtree rooted at $u$,
and let $l_u$ (resp. $r_u$) denote the index of the vertex 
that corresponds to the leftmost (resp. rightmost) leaf of ${\cal T}_u$. 
Let $p_u$ denote the parent of $u$ if $u \neq \rho$. 
We say a node $u \in {\cal T}$ {\it spans} subpath $P_{\ell_u, r_u}$. 
If $P_{\ell_u,r_u} \subseteq P'$ and $P_{\ell_{p_u},r_{p_u}} \not\subseteq P'$, 
node $u$ is called a {\it maximal subpath node} for $P'$.
For each node $u \in {\cal T}$, let $m_u$ be the number of edges in subpath $P_{\ell_u, r_u}$, 
i.e., $m_u = r_u - \ell_u$.
As with a standard segment tree, ${\cal T}$ has the following properties.
\begin{property}\label{property2}
For $i, j \in [1..n]$ with $i<j$, 
the number of maximal subpath nodes for $P_{i,j}$ is $O(\log n)$.
Moreover, 
we can find all the maximal subpath nodes for $P_{i,j}$ 
by walking on ${\cal T}$ from leaf $i$ to leaf $j$ in $O(\log n)$ time.
\end{property}
\begin{property}\label{property1}
If one can construct data structures for each node $u$ of a segment tree ${\cal T}$ in $O(f(m_u))$ time, 
where $f:{\mathbb N} \rightarrow \mathbb{R}$ is some function 
independent of $n$ and bounded below by a linear function asymptotically, i.e., $f(m)=\Omega(m)$, 
then the running time for construction of data structures for every node in ${\cal T}$ is $O(f(n)\log n)$ time in total.
\end{property}

At each node $u \in {\cal T}$, we store four types of the information that depend on 
the indices of the vertices spanned by $u$, i.e., $l_u, \ldots, r_u$. 
We will introduce each type in Section~\ref{sec:data_structure}. 
As will be shown there, the four types of the information at $u \in {\cal T}$ can be constructed in $O(m_u \log m_u)$ time.
Therefore, we can construct ${\cal T}$ in $O(n\log^2 n)$ time by Property~\ref{property1}. 

Recall that 
for $i,j \in [1..n]$ with $i < j$, 
it holds $w'(i,j) = \opt(i,j)$.
We give an outline of the algorithm 
that computes $\opt(i,j)$
only for the non-confluent flow model 
since a similar argument holds even for the confluent flow model with minor modification. 
The main task is to 
find a value $z^*$ that minimizes $\Phi^{i,j}(z)$,
i.e., $\opt(i,j)=\Phi^{i,j}(z^*)$. 
By Lemma~\ref{lem:convex}, such the value $z^*$ is the pseudo-intersection point of $\theta^{i,+}(z)$ and $\theta^{j,-}(z)$ on $[W_i,W_{j-1}]$. 

Before explaining our algorithms, we need introduce the following definition:
\begin{definition}\label{def:subup}
For integers $i,\ell,r \in [1..n]$ with $i < \ell \leq r$, 
we denote by $\theta^{i, +, [\ell..r]}(z)$
the upper envelope of functions $\{\theta^{i, +, h}(z) \mid h \in [\ell..r]\}$, that is, 
\begin{eqnarray}\label{eq:subup_right_define}
\theta^{i, +, [\ell..r]}(z) = \max\{ \theta^{i, +, h}(z) \mid h \in [\ell..r]\}. \notag
\end{eqnarray}
For integers $i,\ell,r \in [1..n]$ with $\ell \leq r < i$,
we denote by $\theta^{i, -, [\ell..r]}(z)$
the upper envelope of functions $\{\theta^{i, -, h}(z) \mid h \in [\ell..r]\}$, that is, 
\begin{eqnarray}\label{eq:subup_left_define}
\theta^{i, -, [\ell..r]}(z) = \max\{ \theta^{i, -, h}(z) \mid h \in [\ell..r]\}. \notag
\end{eqnarray}
\end{definition}

\noindent
{\bf Algorithm for computing $\opt(i,j)$ for given $i,j \in [1..n]$ with $i < j$}
\begin{description}
\item [Phase~{1}:] 
Find a set $U$ of the maximal subpath nodes for $P_{{i+1},{j-1}}$
by walking on segment tree ${\cal T}$ from leaf $i+1$ to leaf $j-1$. 

\item[Phase~{2}:] 
For each $u \in U$, 
compute a real interval ${\cal I}^+_u$ such that 
$\theta^{i,+}(z) = \theta^{i,+,[\ell_u..r_u]}(z)$ holds on any $z \in {\cal I}^+_u$,
and
a real interval ${\cal I}^-_u$ such that 
$\theta^{j,-}(z) = \theta^{j,-,[\ell_u..r_u]}(z)$ holds on any $z \in {\cal I}^-_u$,
both of which are obtained by using information stored at node $u$.
See Section~\ref{app:phase2}. 

\item[Phase~{3}:] Compute the pseudo-intersection point $z^*$ 
of $\theta^{i,+}(z)$ and $\theta^{j,-}(z)$ on $[W_i,W_{j-1}]$ 
by using real intervals obtained in Phase~{2}.
See Section~\ref{app:phase3}.
\item[Phase~{4}:] 
Compute $\opt(i,j) = \Phi^{i,j}(z^*)$ as follows:
By formula~\eqref{eq:at_subpath}, we have 
\begin{eqnarray*}
\Phi^{i,j}(z^*) &=& \int_{0}^{z^*} \theta^{i, +}(t)dt + \int_{z^*}^{W_n} \theta^{j, -}(t)dt \\
&=& 
\sum_{u \in U}\left\{
\int_{{\cal I}^+_u \cap [0, z^*]} \theta^{i, +, [\ell_u..r_u]}(t)dt + 
\int_{{\cal I}^-_u \cap [z^*, W_n]} \theta^{j, -, [\ell_u..r_u]}(t)dt
\right\}. 
\end{eqnarray*}
For each $u \in U$, we compute integrals $\int\theta^{i, +, [\ell_u..r_u]}(t)dt$ and 
$\int\theta^{j, -, [\ell_u..r_u]}(t)dt$ 
by using the information stored at $u$. 
See Section~\ref{app:phase4} for the details.
\end{description}
For the cases of $i=0$ or $j=n+1$, 
we can also compute $w'(0,j) = \Phi^{j, -}(0)$ and $w'(i,n+1) = \Phi^{i, +}(W_n)$ by the same operations except for Phase~{3}.

We give the following lemma about the running time of the above algorithm 
for the case of general edge capacities. 
See Section~\ref{app:gen_edge_capacity} for the proof.
\begin{lemma}[Key lemma for general capacity]\label{lem:general_query}
Let us suppose that a segment tree ${\cal T}$ is available. 
Given two integers $i,j \in [0..n+1]$ with $i<j$,
one can compute a value $w'(i,j)$ 
in $O(\log^3 n)$ time 
for the confluent/non-confluent flow model.
\end{lemma}

Recalling that the running time for construction of data structure ${\cal T}$ is $O(n \log^2 n)$, 
Lemmas~\ref{lem:klink},~\ref{lem:concave_monge} and~\ref{lem:general_query} imply the following main theorem. 
\begin{theorem}[Main theorem for general capacity]\label{thm:general}
Given a dynamic flow path network $\mathcal{P}$, there exists an algorithm that finds an optimal $k$-sink 
under the confluent/non-confluent flow model
in time $\min\{ O(kn\log^3n),n 2^{O(\sqrt{\log k \log\log n})}$ $\log^3n \}$.
\end{theorem}

When the capacities of ${\cal P}$ are uniform, 
we can improve the running time for computing $w'(i,j)$
to $O(\log^2 n)$ time with minor modification. 
See Section~\ref{app:uniform_edge_capacity}.

\begin{lemma}[Key lemma for uniform capacity]\label{lem:uniform_query}
Let us suppose that a segment tree ${\cal T}$ is available. 
Given two integers $i,j \in [1..n]$ with $i<j$, 
one can compute a value $w'(i,j)$ in $O(\log^2 n)$ time 
for the confluent/non-confluent flow model when the capacities are uniform. 
\end{lemma} 
\begin{theorem}[Main theorem for uniform capacity]\label{thm:uniform}
Given a dynamic flow path network $\mathcal{P}$ with a uniform capacity, 
there exists an algorithm that finds an optimal $k$-sink
under the confluent/non-confluent flow model in time\\
$\min\{ O(kn\log^2n), n 2^{O(\sqrt{\log k \log\log n})} \log^2n \}$.
\end{theorem}

\section{Data Structures Associated with Nodes of ${\cal T}$}\label{sec:data_structure}

In the rest of the paper, we introduce novel data structures 
associated with each node $u$ of segment tree ${\cal T}$, 
which are used to compute $\opt(i,j)$ in $O({\rm poly}\log n)$ time. 
Note that our data structures generalize the {\it capacities and upper envelopes tree (CUE tree)} provided by Bhattacharya et al.~\cite{BhattacharyaGHK17}.

Recall the algorithm for computing $\opt(i,j)$ shown in Section~\ref{sec:algorithm}.
To explain the data structures, 
let us see more precisely how the algorithm performs in Phase~2.
Confirm that for $z \in [W_i, W_{j-1}]$,
it holds $\theta^{i,+}(z) = \max\{\theta^{i,+,[\ell_u..r_u]}(z) \mid u \in U\}$, where $U$ is a set of the maximal subpath nodes for $P_{{i+1},{j-1}}$.
Let us focus on function $\theta^{i,+,[\ell_u..r_u]}(z)$ for a node $u \in U$
only on interval $(W_{\ell_u-1},W_n]$
since
it holds $\theta^{i,+,[\ell_u..r_u]}(z)=0$ if $z \le W_{\ell_u-1}$.
Interval $(W_{\ell_u-1},W_n]$ consists of three left-open-right-closed intervals ${\cal J}^{+}_{u,1}$, ${\cal J}^{+}_{u,2}$ and ${\cal J}^{+}_{u,3}$ 
that satisfy the following conditions:
(i) For $z \in {\cal J}^{+}_{u,1}$, $\theta^{i,+,[\ell_u..r_u]}(z)=\theta^{i,+,\ell_u}(z)$.
(ii) For $z \in {\cal J}^{+}_{u,2}$, $\theta^{i,+,[\ell_u..r_u]}(z)=\theta^{i,+,[\ell_u+1..r_u]}(z)$ and
its slope is $1/{C_{i,\ell_u}}$.
(iii) For $z \in {\cal J}^{+}_{u,3}$, $\theta^{i,+,[\ell_u..r_u]}(z)=\theta^{i,+,[\ell_u+1..r_u]}(z)$ and
its slope is greater than $1/{C_{i,\ell_u}}$.
See also Fig.~\ref{fig:5}.
Thus in Phase~2, 
the algorithm computes ${\cal J}^{+}_{u,1}$, ${\cal J}^{+}_{u,2}$ and ${\cal J}^{+}_{u,3}$ for all $u \in U$,
and combines them one by one 
to obtain 
intervals ${\cal I}^+_u$ for all $u \in U$.
To implement these operations efficiently, 
we construct some data structures at each node $u$ of ${\cal T}$.
To explain the data structures stored at $u$,
we introduce the following definition:
\begin{definition}\label{def:bar_subup} 
For integers $i,\ell,r \in [1..n]$ with $i < \ell \leq r$ and a positive real $c$,
let $\bar{\theta}^{i, +, [\ell..r]}(c,z) = \max\{ \bar{\theta}^{i, +, h}(c,z) \mid h \in [\ell..r]\}$, where
\begin{eqnarray}\label{eq:bar_subfunction_right}
\bar{\theta}^{i,+,j}(c, z) = 
\left\{
\begin{array}{ll}
0 & \text{if } z \leq W_{j-1}, \\
\frac{z - W_{j-1}}{c} + \tau\cdot L_{i,j} & \text{if } z > W_{j-1}.
\end{array}
\right.
\end{eqnarray} 
For integers $i,\ell,r \in [1..n]$ with $\ell \leq r < i$ and a positive real $c$,
let $\bar{\theta}^{i, -, [\ell..r]}(c,z) = \max\{ \bar{\theta}^{i, -, h}(c,z) \mid h \in [\ell..r]\}$, where
\begin{eqnarray}\label{eq:bar_subfunction_left}
\bar{\theta}^{i,-,j}(c, z) = 
\left\{
\begin{array}{ll}
\frac{W_{j} - z}{c} + \tau\cdot L_{j,i} & \text{if } z < W_j, \\
0 & \text{if } z \geq W_j.
\end{array}
\right.
\end{eqnarray} 
\end{definition}
We can see that 
for $z \in {\cal J}^{+}_{u,2}$,
$\theta^{i,+,[\ell_u+1..r_u]}(z)=\bar{\theta}^{\ell_u, +, [\ell_u+1..r_u]}(C_{i,\ell_u},z)+\tau \cdot L_{i,\ell_u}$,
and
for $z \in {\cal J}^{+}_{u,3}$,
$\theta^{i,+,[\ell_u+1..r_u]}(z)=\theta^{\ell_u,+,[\ell_u+1..r_u]}(z)+\tau \cdot L_{i,\ell_u}$.
We then store at $u$ of ${\cal T}$ the information for computing 
in $O({\rm poly}\log n)$ time
$\theta^{\ell_u,+,[\ell_u+1..r_u]}(z)$ for any $z \in [0,W_n]$ 
as TYPE~I, and also
one for computing in $O({\rm poly}\log n)$ time
$\bar{\theta}^{\ell_u,+,[\ell_u+1..r_u]}(c, z)$
for any $c>0$ and any $z \in [0,W_n]$ as TYPE~III. 

\begin{figure}[tb]
  \begin{center}
    \includegraphics[width=7cm,pagebox=cropbox,clip]{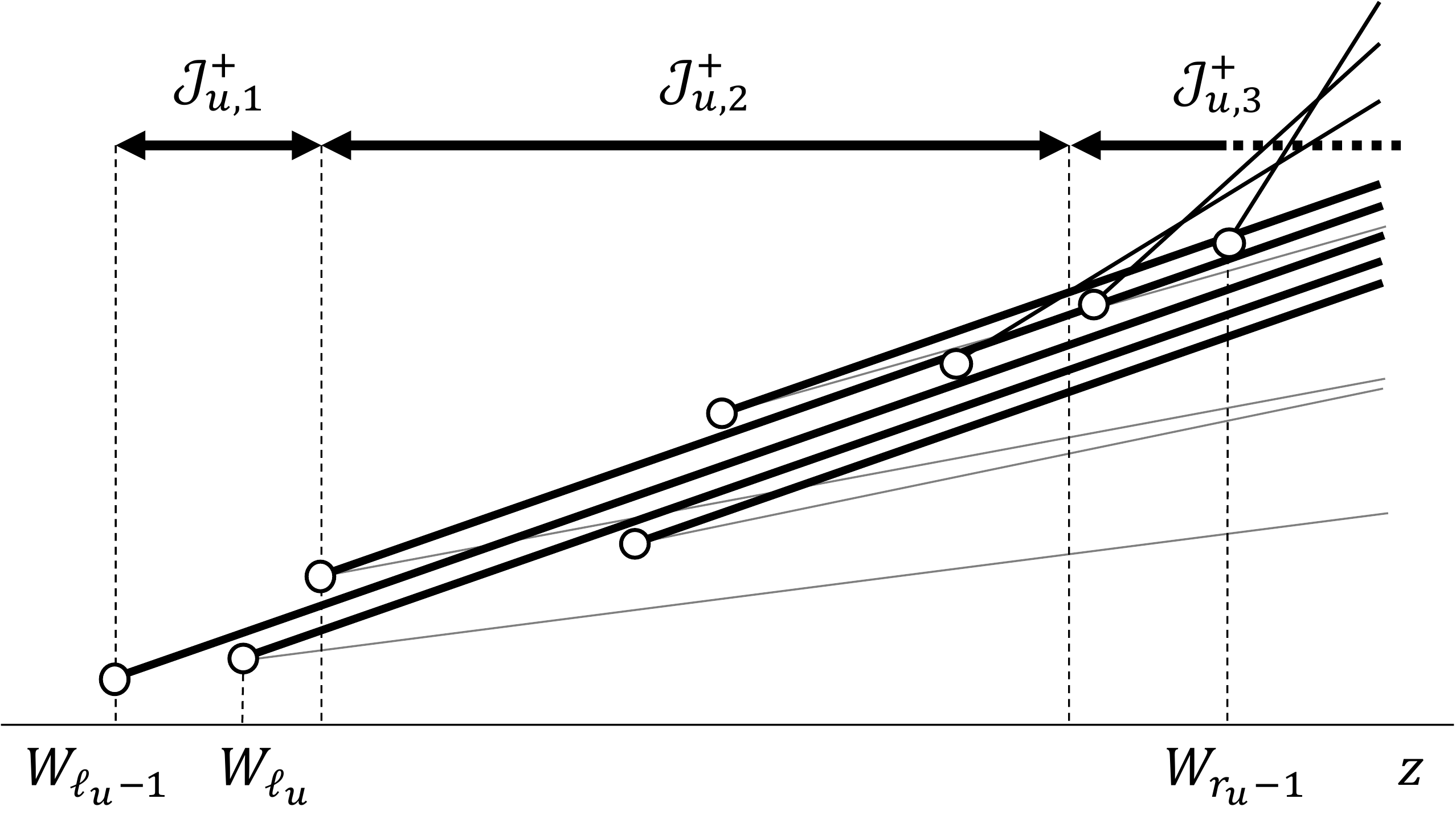}
    \caption{Illustration of ${\cal J}^{+}_{u,1}$, ${\cal J}^{+}_{u,2}$ and ${\cal J}^{+}_{u,3}$.
    The thick half lines have the same slope of $1/C_{i,\ell_u}$,
    the gray half lines have slopes $\le 1/C_{i,\ell_u}$,
    and the regular half lines have slopes $>1/C_{i,\ell_u}$.
    The upper envelope of all the thick half lines and the regular half lines is function $\theta^{i,+,[\ell_u..r_u]}(z)$.
    }
    \label{fig:5}
  \end{center}
\end{figure}

In Phase~4, the algorithm requires computing integrals 
$\int_0^z \theta^{\ell_u, +, [\ell_u+1.. r_u]}(t)dt$ for any $z \in [0,W_n]$, and 
$\int_0^z \bar{\theta}^{\ell_u,+,[\ell_u+1..r_u]}(c, t) dt$ for any $c>0$ and any $z \in [0,W_n]$, 
for which the information is stored at each $u \in {\cal T}$ as TYPEs~II and IV, respectively.

In a symmetric manner, we also store at each $u \in {\cal T}$ the information for computing
$\theta^{r_u, -, [\ell_u.. r_u-1]}(z)$, 
$\int_z^{W_n} \theta^{r_u, -, [\ell_u.. r_u-1]}(t)dt$,
$\bar{\theta}^{r_u,-, [\ell_u.. r_u-1]}(c, z)$, and
$\int_z^{W_n}\bar{\theta}^{r_u,-,[\ell_u..r_u-1]}(c, t)dt$ 
as TYPEs~{I}, II, III, and~IV, respectively.

Let us introduce what information is stored as TYPEs~{I}--{IV} at $u \in {\cal T}$. 

\noindent
{\bf TYPE~I.} We give the information only for computing $\theta^{\ell_u,+,[\ell_u+1..r_u]}(z)$ 
stored at $u \in {\cal T}$ as TYPE~I 
since the case for $\theta^{r_u, -, [\ell_u..r_u-1]}(z)$ is symmetric. 
By Definition~\ref{def:subup}, 
the function $\theta^{\ell_u, +, [\ell_u+1.. r_u]}(z)$ is the upper envelope of $m_u$ functions 
$\theta^{\ell_u, +, h}(z)$ for $h \in [\ell_u+1..r_u]$.
Let ${\cal B}^{u,+} = (b^{u,+}_1 = 0, b^{u,+}_2, \ldots, b^{u,+}_{N^{u,+}} = W_n)$ denote
a sequence of breakpoints of $\theta^{\ell_u, +, [\ell_u+1.. r_u]}(z)$, where $N^{u,+}$ is the number of breakpoints.
For each $p \in [1..N^{u,+} - 1]$,
let $H^{u,+}_p \in [\ell_u+1..r_u]$ such that 
$\theta^{\ell_u, +, [\ell_u+1.. r_u]}(z)=\theta^{\ell_u, +, H^{u,+}_p}(z)$ holds for any $z \in (b^{u,+}_{p}, b^{u,+}_{p+1}]$.
As TYPE~{I}, each node $u \in {\cal T}$ is associated with following two lists:

1. Pairs of breakpoint $b^{u,+}_p$ and value 
$\theta^{\ell_u, +, [\ell_u+1.. r_u]}(b^{u,+}_p)$, and

2. Pairs of range $(b^{u,+}_{p}, b^{u,+}_{p+1}]$ and index $H^{u,+}_p$.\\
Note that the above lists can be constructed in $O(m_u \log m_u)$ time for each $u \in {\cal T}$ as follows. 
Applying the result shown by Hershberger~\cite{Hershberger89}, we construct all lists for TYPE~I, efficiently.
\begin{lemma}[\cite{Hershberger89}]\label{lem:Hershberger}
When we have $n$ line segments, 
there exists an algorithm that computes the upper envelope of these segments in $O(n \log n)$ time.
\end{lemma}
By formulae~{(\ref{eq:subfunction_right})} and~{(\ref{eq:subfunction_left})},
for any $i$ and $j$, functions $\theta^{i, +, j}$ and $\theta^{i, -, j}$ consist of at most two line segments.
Thus, since $\theta^{\ell_u, +, [\ell_u+1.. r_u]}(z)$ and $\theta^{r_u, -, [\ell_u.. r_u-1]}(z)$ are upper envelopes of at most $2m_u$ line segments.
By Lemma~\ref{lem:Hershberger}, we can obtain these upper envelopes and the above four lists for each node $u \in {\cal T}$ in $O( m_u \log m_u )$ time.
In total, we can obtain the whole information of TYPE~I of ${\cal T}$ in time $O(n\log^2 n)$ by Property~\ref{property1}.

We now give the application of TYPE~I. 
\begin{lemma}[Query with TYPE~I]\label{lem:type1}
Suppose that TYPE~I of ${\cal T}$ is available. 
Given a node $u \in {\cal T}$ and a real value $z \in [0,W_n]$,
we can obtain

\noindent
(i) index $H \in [\ell_u+1..r_u]$ such that $\theta^{\ell_u, +, H}(z) = \theta^{\ell_u,+,[\ell_u+1.. r_u]}(z)$, and

\noindent
(ii)
index $H \in [\ell_u..r_u-1]$ such that 
$\theta^{r_u, -, H}(z) = \theta^{r_u, -, [\ell_u.. r_u-1]}(z)$\\
in time $O(\log n)$ respectively. 
Furthermore, if the capacities of $\mathcal{P}$ are uniform
and $z \notin [W_{\ell_u},W_{r_u-1}]$, we can obtain
the above indices in time $O(1)$.
\end{lemma}
\begin{proof}
We give the proof for the cases for $\theta^{\ell_u, +, [\ell_u+1.. r_u]}(z)$.
Cases for $\theta^{r_u, -, [\ell_u.. r_u-1]}(z)$ can be shown in a similar way.

First, we give the operations for case~{(i)}: 
Find an integer $p$ such that $z \in (b^{u,+}_p, b^{u,+}_{p+1}]$ holds 
in $O(\log n)$ time by the binary search, 
which compares a given $z$ and a breakpoint $b^{u,+}_p$ with list~{TYPE~I-1}. 
Then, we obtain an index $H^{u,+}_p$ by list~{TYPE~I-2}.

Next, we consider the case that the dynamic flow path network $\mathcal{P}$ has uniform capacity $c$.
We recall the definition of $\theta^{i, +, j}(z)$ when every capacity is $c$
from equation~{(\ref{eq:subfunction_right})}.
For any $i,j \in [1..n]$ with $i < j$, we have
\begin{eqnarray*}
\theta^{i, +, j}(z) = 
\left\{
\begin{array}{ll}
0 & \text{if } z \leq W_{j-1}, \\
\frac{z - W_{j-1}}{c} + \tau\cdot L_{i,j} & \text{if } z > W_{j-1}.
\end{array}
\right.
\end{eqnarray*}
Because every function $\theta^{i, +, j}$ takes zero or has slope $1/c$,
all breakpoints $b^{u, +}_{p}$ of ${\cal B}^{u,+}$ should be 
$0, W_{\ell_u}, W_{\ell_u+1}, \ldots, W_{r_u}$ for $p = 1, \ldots, N^{u,+} - 1$.
This implies that a range of $[b^{u,+}_1(=0), b^{u,+}_2]$ with the smallest two breakpoints
contains the range $[0, W_{\ell_u}]$.
Thus, for $z \in [0, W_{\ell_u}]$, we have $\theta^{\ell_u, +, H^{u,+}_{1}}(z) = \theta^{\ell_u, +, [\ell_u+1.. r_u]}(z) = 0$.
Similarly, a range of $[b^{u,+}_{N_u-1}, b^{u,+}_{N_u}(= W_n)]$ with 
the largest two breakpoints contains the range $(W_{r_u}, W_n]$
and we have $\theta^{\ell_u, +, H^+_{N^{u,+}-1}}(z) = \theta^{\ell_u, +, [\ell_u+1.. r_u]}(z)$ for $z \in (W_{r_u}, W_{n}]$.
We can obtain $H^{u,+}_{1}$ and $H^+_{N^{u,+}-1}$ in $O(1)$ time and the proof is complete.
\end{proof}



\noindent
{\bf TYPE~II.} We give the information only for computing 
$\int_0^z \theta^{\ell_u, +, [\ell_u+1.. r_u]}(t)dt$ 
stored at $u \in {\cal T}$ as TYPE~{II} 
since the case for $\int_z^{W_n} \theta^{r_u, -, [\ell_u.. r_u-1]}(t)dt$ is symmetric.
Each node $u \in {\cal T}$ contains 
a list of all pairs of breakpoint $b^{u,+}_p$ and value $\int_{0}^{b^{u,+}_p}\theta^{\ell_u, +, [\ell_u+1.. r_u]}(t)dt$. 
We show that these lists can be constructed in $O(m_u)$ time for each $u \in {\cal T}_u$.
Let us consider the calculation of $\int_{0}^{b^{u,+}_{p}}\theta^{\ell_u, +, [\ell_u+1.. r_u]}(t)dt$.
By the list TYPE~{I-2}, we have $\theta^{\ell_u, H^{u,+}_{p}, +} (z)$ for $p = 1, \ldots, N^{u,+}$.
By an elementary calculation, we calculate a value $\int_{b^{u,+}_{p}}^{b^{u,+}_{p+1}} \theta^{\ell_u, +, H^{u,+}_{p+1}} (t)dt$
in $O(1)$ time for each $p$.
Because we have the following relation
\begin{eqnarray*}
\int_{0}^{b^{u,+}_{p+1}}\theta^{\ell_u, +, [\ell_u+1.. r_u]}(t)dt = 
\int_{0}^{b^{u,+}_p}\theta^{\ell_u, +, [\ell_u+1.. r_u]}(t)dt + \int_{b^{u,+}_{p}}^{b^{u,+}_{p+1}} \theta^{\ell_u, H^{u,+}_{p+1}, +} (t)dt,
\end{eqnarray*}
for $p = 1, \ldots, N^{u, +} -1$,
we can obtain $\int_{0}^{b^{u,+}_{p}}\theta^{\ell_u, +, [\ell_u+1.. r_u]}(t)dt$ for all $p$ in $O(m_u)$ time by
adding calculated values.
In a similar way, we also construct the list of $\int_{b^{u,-}_p}^{W_n}\theta^{r_u, -, [\ell_u.. r_u-1]}(t)dt$ 
in $O(m_u)$ time. 
In total, we obtain the whole information of TYPE~{II} of ${\cal T}$ in time $O(n\log n)$ by Property~\ref{property1}. 

We give the application of TYPEs~{I and II}. 
\begin{lemma}[Query with TYPEs~{I and II}]\label{lem:type2}
Suppose that TYPEs~{I and II} of ${\cal T}$ is available.
Given a node $u \in {\cal T}$ and a real value $z \in [0,W_n]$, we can obtain

\noindent
(i) value $\int_{0}^{z} \theta^{\ell_u, +, [\ell_u+1.. r_u]}(t)dt$, and
(ii) value $\int_{z}^{W_n} \theta^{r_u, -, [\ell_u.. r_u-1]}(t)dt$
in time $O(\log n)$ respectively. 
Furthermore, if the capacities of $\mathcal{P}$ are uniform
and $z \notin [W_{\ell_u},W_{r_u-1}]$, 
we can obtain the above values in time $O(1)$.
\end{lemma}
\begin{proof}
We give the proof for the cases for 
$\int_{0}^{z} \theta^{\ell_u, +,  [\ell_u+1.. r_u]}(t)dt$. Cases for $\int_{0}^{z} \theta^{r_u, -, [\ell_u.. r_u-1]}(t)dt$ can be shown by a similar way.

First, we consider case~{(i)}. 
By Lemma~\ref{lem:type1}-(i), 
we obtain an index $H^{u,+}_p$ such that given $z$ is contained in $(b^{u,+}_p, b^{u,+}_{p+1}]$ and $\theta^{\ell_u, +,  [\ell_u+1.. r_u]}(z') = \theta^{\ell_u, +, H^{u,+}_p} (z')$ for any $z' \in (b^{u,+}_p, b^{u,+}_{p+1}]$ in $O(\log n)$ time.
Using this index $H^{u,+}_p$, we have
\[
\int_{0}^{z} \theta^{\ell_u, +,  [\ell_u+1.. r_u]}(t)dt = \int_{0}^{b^{u,+}_p} \theta^{\ell_u, +,  [\ell_u+1.. r_u]}(t)dt +
\int_{b^{u,+}_{p}}^{z} \theta^{\ell_u, +, H^{u,+}_p} (t)dt.
\]
We obtain a value of the first term $\int_{0}^{b^{u,+}_p} \theta^{\ell_u, +,  [\ell_u+1.. r_u]}(t)dt$ by the binary search on the list of TYPE~{II} in $O(\log n)$ time. 
The second term $\int_{b^{u,+}_{p}}^{z} \theta^{\ell_{u}, +, H^{u,+}_p} (t)dt$ 
is calculated by an elementary calculation in $O(1)$ time, 
since we know the function $\theta^{\ell_u, +, H^{u,+}_p} (z)$.
Thus, the proof of statement~{(i)} is completed.

Next, we consider the case that the capacities of $\mathcal{P}$ are uniform in $c$. 
By the proof of Lemma~\ref{lem:type1}, we know that $H^{u,+}_p = 1$ or $N^{u,+}-1$ without the binary search. 
When $H^{u,+}_p = N^{u,+}-1$, we have
\[
\int_{0}^{z} \theta^{\ell_u, +,  [\ell_u+1.. r_u]}(t)dt 
= \int_{0}^{b^{u,+}_{N^{u,+}-1}} \theta^{\ell_u, +,  [\ell_u+1.. r_u]}(t)dt
+ \int_{b^{u,+}_{p}}^{z} \theta^{\ell_u, +, H^{u,+}_p} (t)dt.
\]
Thus, we can evaluate both terms in $O(1)$ time and the proof is complete.
\end{proof}

\noindent
{\bf TYPE III.} 
We give the information only for computing 
$\bar{\theta}^{\ell_u,+,[\ell_u+1.. r_u]}(c, z)$ 
stored at $u \in {\cal T}$ as TYPE~{III} 
since the case for $\bar{\theta}^{r_u,-, [\ell_u.. r_u-1]}(c, z)$ is symmetric.
Note that it is enough to prepare for the case of $z \in (W_{\ell_u},W_{r_u}]$
since it holds that 
$\bar{\theta}^{\ell_u,+,[\ell_u+1.. r_u]}(c, z) = 0$ for $z \in [0, W_{\ell_u}]$ and
\[
\bar{\theta}^{\ell_u,+,[\ell_u+1.. r_u]}(c, z) = \bar{\theta}^{\ell_u,+,[\ell_u+1.. r_u]}(c, W_{r_u}) + \frac{z-W_{r_u}}{c}
\]
for $z \in (W_{r_u}, W_n]$, of which the first term is obtained by prepared information with 
$z = W_{r_u}$ and the second term is obtained by elementally calculation.

For each $u \in {\cal T}$, 
we construct a {\em persistent segment tree} as TYPE~{III}.
Referring to formula~\eqref{eq:bar_subfunction_right}, 
each function $\bar{\theta}^{\ell_u,+,j}(c, z)$ for $j \in [l_u+1..r_u]$ 
is linear in $z \in (W_{j-1},W_n]$ with the same slope $1/c$. 
Let us make parameter $c$ decrease from $\infty$ to $0$, 
then all the slopes $1/c$ increase from $0$ to $\infty$. 
As $c$ decreases, the number of subfunctions that consist of $\bar{\theta}^{\ell_u,+,[\ell_u+1.. r_u]}(c, z)$ 
also decreases one by one from $m_u$ to $1$. 
Let $c^{u,+}_{h}$ be a value $c$ at which the number of subfunctions of $\bar{\theta}^{\ell_u,+,[\ell_u+1.. r_u]}(c, z)$ becomes $m_u-h$ while $c$ decreases. 
Note that we have $\infty = c^{u,+}_{0} > c^{u,+}_{1} > \cdots > c^{u,+}_{m_u-1} > 0$.
Let us define indices $j^{h}_1, \ldots, j^{h}_{m_u-h}$ with $l_u+1 = j^h_1 < \cdots < j^h_{m_u-h} \le r_u$ corresponding to the subfunctions of $\bar{\theta}^{\ell_u,+,[\ell_u+1.. r_u]}(c^{u,+}_{h}, z)$, 
that is, for any integer $p \in [1..m_u-h]$, we have
\begin{equation}\label{eq:j_u_h}
\bar{\theta}^{\ell_u,+,[\ell_u+1.. r_u]}(c^{u,+}_{h}, z) = \bar{\theta}^{\ell_u,+,j^h_p}(c^{u,+}_{h}, z) \quad \text{if} \ z \in (W_{j^{h}_p-1},W_{j^{h}_{p+1}-1}],
\end{equation}
where $j^{h}_{m_u-h+1}-1={r_u}$. 
We give the following lemma about the property of $c^{u,+}_{h}$. 
\begin{lemma}\label{lem:capacity}
For each node $u \in {\cal T}$, 
all values $c^{u,+}_{1}, \ldots, c^{u,+}_{m_u-1}$ can be computed 
in $O(m_u \log m_u)$ time.
\end{lemma}
\begin{proof}
Suppose that $c^{u,+}_{0} = \infty, \ldots, c^{u,+}_{h-1}$ with $1 \le h \le m_u-2$ have been computed so far.
Recall that $\bar{\theta}^{\ell_u,+,[\ell+1, r_u]}(c^{u,+}_{h-1}, z)$ is a piecewise linear function with $m_u-h+1$ subfunctions on $(W_{\ell_u},W_n]$. 
Let indices $j^{h-1}_1, \ldots, j^{h-1}_{m_u-h+1}$ with $l_u+1 = j^{h-1}_1 < \cdots < j^{h-1}_{m_u-h+1} \le r_u$ correspond to the subfunctions of $\bar{\theta}^{\ell_u,+,[\ell_u+1.. r_u]}(c^{u,+}_{h-1}, z)$,
that is, for any integer $p$ with $1 \le p \le m_u-h+1$, we have
\begin{equation*}
\bar{\theta}^{\ell_u,+,[\ell_u+1.. r_u]}(c^{u,+}_{h-1}, z) = \bar{\theta}^{\ell_u,+,j^{h-1}_p}(c^{u,+}_{h-1}, z) \quad \text{if} \ z \in (W_{j^{h-1}_p-1},W_{j^{h-1}_{p+1}-1}],
\end{equation*}
where $j^{h-1}_{m_u-h+2}-1=n$.

For an integer $p$ with $1 \le p \le m_u-h$, 
let $c_{h,p}$ be a value such that
two functions $\bar{\theta}^{\ell_u,+,j^{h-1}_p}( c_{h,p}, z)$ and  $\bar{\theta}^{\ell_u,+,j^{h-1}_{p+1}}( c_{h,p}, z)$ have an overlap on $(W_{j^{h-1}_{p+1}-1},W_{j^{h-1}_{p+2}-1}]$. 
This means that for any $z \in (W_{j^{h-1}_{p+1}-1},W_{j^{h-1}_{p+2}-1}]$, we have
\begin{eqnarray*}
\frac{z - W_{j^{h-1}_p-1}}{c^{u,+}_{h,p}} + \tau\cdot L_{\ell_u,j^{h-1}_p} &=& 
\frac{z - W_{j^{h-1}_{p+1}-1}}{c^{u,+}_{h,p}} + \tau\cdot L_{\ell_u,j^{h-1}_{p+1}} \end{eqnarray*}
and then
\begin{eqnarray}
c^{u,+}_{h,p} &=& \frac{W_{j^{h-1}_{p+1}-1} - W_{j^{h-1}_p-1}}{\tau\cdot L_{j^{h-1}_p,j^{h-1}_{p+1}}}.
\end{eqnarray}
By the definition of $c^{u,+}_{h}$, we have
\begin{eqnarray}\label{eq:c_u+_h}
c^{u,+}_{h}
= \max_{1 \le p \le m_u-h}
\left\{ c^{u,+}_{h,p} \right\}.
\end{eqnarray}

In order to obtain $c^{u,+}_{h}$, we construct a {\it max-heap} 
that contains all values $c^{u,+}_{h,p}$ for $1 \le p \le m_u-h$
as follows:

We first construct a max-heap with all $m_u - 1$ values 
$c^{u,+}_{1,p} = \frac{W_{\ell_u + p + 1} - W_{\ell_u + p}}{\tau \cdot 
L_{\ell_u + p,\ell_u + p + 1}}
= \frac{ w_{\ell_u + p + 1} }{\tau \cdot \ell_{\ell_u + p}}$ 
for $1 \le p \le m_u-1$ in $O(m_u \log m_u)$ time 
and then obtain a value $c^{u,+}_{1}$. 

Suppose that values $c^{u,+}_{h,p}$ for $1 \le p \le m_u-h$ are stored in a max-heap. 
We then immediately obtain $c^{u,+}_{h}$. 
We update the max-heap as follows: 
Letting $p'$ be the maximizer in \eqref{eq:c_u+_h}, 
delete two values $c^{u,+}_{h,p'} = c^{u,+}_{h}$ and $c^{u,+}_{h,p'+1}$, and 
insert a value
\begin{eqnarray}
c^{u,+}_{h+1,p'} = 
\frac{W_{j^{h-1}_{p'+2}-1} - W_{j^{h-1}_{p'}-1}}{\tau\cdot L_{j^{h-1}_{p'},j^{h-1}_{p'+2}}}.
\end{eqnarray}
Note that each operation can be done in $O(\log m_u)$ time.
For $p \neq p'$, we have 
\begin{eqnarray}\label{eq:c_u+_h+1}
c^{u,+}_{h+1,p} =
\left\{
\begin{array}{ll}
c^{u,+}_{h,p} & \text{if } p < p', \\
c^{u,+}_{h,p+1} & \text{if } p > p'.
\end{array}
\right.
\end{eqnarray} 
Thus, the updated max-heap contains all values $c^{u,+}_{h+1,p}$ 
for $1 \le p \le m_u-h-1$ and obtain the value $c^{u,+}_{h+1}$.
We repeat the above mentioned updates $O(m_u)$ times to obtain $c^{u,+}_{2}, \ldots, c^{u,+}_{m_u-1}$,
which requires $O(m_u \log m_u)$ time in total.
This completes the proof.
\end{proof}

By the above argument, 
while $c \in (c^{u,+}_{h},c^{u,+}_{h-1}]$ with some $h \in [1..m_u]$
(where $c^{u,+}_{m_u}=0$),
the representation of $\bar{\theta}^{\ell_u,+,[\ell_u+1.. r_u]}(c, z)$ (with $m_u-h+1$ subfunctions) remains the same.
Our fundamental idea is to consider
segment trees corresponding to each interval $(c^{u,+}_{h},c^{u,+}_{h-1}]$ with $h \in [1..m_u]$, and construct a persistent data structure for such the segment trees.

First of all, we introduce a segment tree $T_h$ with root $\rho_h$ 
to compute $\bar{\theta}^{\ell_u, +,[\ell_u+1.. r_u]}(c, z)$ for $c \in (c^{u,+}_{h},c^{u,+}_{h-1}]$ with $h \in [1..m_u]$. 
Tree $T_h$ contains $m_u$ leaves labeled as $l_u+1, \ldots, r_u$.
Each leaf $j$ corresponds to interval $(W_{j-1}, W_{j}]$. 
For a node $\nu \in T_h$, 
let $\ell_{\nu}$ (resp. $r_{\nu}$) denote 
the label of the leftmost (resp. rightmost) leaf of the subtree rooted at ${\nu}$. 
Let $p_{\nu}$ denote the parent of $\nu$ if $\nu \neq \rho_h$. 
We say a node $\nu \in T_h$ {\it spans} an interval $(W_{\ell_{\nu}-1}, W_{r_{\nu}}]$.
For some two integers $i,j \in [\ell_u+1..r_u]$ with $i<j$, 
if $(W_{\ell_{\nu}-1}, W_{r_{\nu}}] \subseteq (W_{i-1}, W_j]$ 
and $(W_{\ell_{p_{\nu}-1}}, W_{r_{p_{\nu}}}] \not \subseteq (W_{i-1}, W_j]$, 
then $\nu$ is called a {\it maximal subinterval node} for $(W_{i-1}, W_j]$. 
A segment tree $T_h$ satisfies the following property 
similar to Property~\ref{property2}:
For any two integers $i,j \in [\ell_u+1..r_u]$ with $i<j$, 
the number of maximal subinterval nodes in $T_h$ for $(W_{i-1}, W_j]$ is $O(\log m_u)$. 
For each $p \in [1..m_u-h+1]$, 
we store function $\bar{\theta}^{\ell_u,+,j^{h-1}_p}(c, z)$ 
at all the maximal subinterval nodes for interval $(W_{j^{h-1}_p-1}, W_{j^{h-1}_{p+1}-1}]$,
which takes $O(m_u\log m_u)$ time by 
the property.
The other nodes in $T_h$ contains {\sf NULL}. 

If we have $T_h$, 
for given $z \in (W_{\ell_u},W_{r_u}]$ and $c \in (c^{u,+}_{h},c^{u,+}_{h-1}]$, 
we can compute value $\bar{\theta}^{\ell_u,+,[\ell_u+1.. r_u]}(c, z)$
in time $O(\log m_u)$ as follows:
Starting from root $\rho_h$, 
go down to a child such that 
its spanned interval contains $z$ 
until we achieve a node that 
contains some function $\bar{\theta}^{\ell_u,+,j}(c, z)$ (not {\sf NULL}). 
Now, we know 
$\bar{\theta}^{\ell_u,+,[\ell_u+1.. r_u]}(c, z) = \bar{\theta}^{\ell_u,+,j}(c, z)$, 
which can be computed by elementally calculation. 
 
If we explicitly construct $T_h$ for all $h \in [1..m_u]$, it takes $O(m_u^2\log m_u)$ time for each node $u \in {\cal T}$, which implies by Property~\ref{property1} that $O(n^2\log^2 n)$ time is required in total. 
However, using the fact that 
$T_h$ and $T_{h+1}$ are almost same except for at most $O(\log m_u)$ nodes, 
we can construct a persistent segment tree in $O(m_u\log m_u)$ time, 
in which we can search as if all of $T_h$ are maintained as follows. 

We first construct a segment tree $T_1$ 
such that all inner nodes contain {\sf NULL} 
and each leaf $j \in [\ell_u+1..r_u]$ has function $\bar{\theta}^{\ell_u,+,j}(c, z)$ in $O(m_u)$ time. 

We update this tree as follows:
Let us suppose that a persistent segment tree contains the information of $T_h$.
We update this persistent segment tree by adding the information of $T_{h+1}$.
Let $p'$ be the maximizer in \eqref{eq:c_u+_h},
i.e., when $c = c^{u,+}_h$, 
$\bar{\theta}^{\ell_u,+,j^{h-1}_{p'}}(c, z)$ and $\bar{\theta}^{\ell_u,+,j^{h-1}_{p'+1}}(c, z)$ come to overlap each other on $(W_{j^{h-1}_{p'+1}-1},W_{j^{h-1}_{p'+2}-1}]$.
Recall that in $T_h$, 
$\bar{\theta}^{\ell_u,+,j^{h-1}_{p'}}(c, z)$ and  $\bar{\theta}^{\ell_u,+,j^{h-1}_{p'+1}}(c, z)$
are stored at all the maximal subinterval nodes for
intervals $(W_{j^{h-1}_{p'}-1}, W_{j^{h-1}_{p'+1}-1}]$ and $(W_{j^{h-1}_{p'+1}-1}, W_{j^{h-1}_{p'+2}-1}]$,
respectively.
Therefore, if we delete these information and store 
$\bar{\theta}^{\ell_u,+,j^{h-1}_{p'}}(c, z)$ 
at all the maximal subinterval nodes for interval $(W_{j^{h-1}_{p'}-1}, W_{j^{h-1}_{p'+2}-1}]$,
then we obtain $T_{h+1}$.
Instead,
we prepare a copy of subtree $T'$ of $T_{h}$, say $T'_{copy}$,
where $T'$ is the minimal subtree containing all the maximal subinterval nodes in $T_{h}$ for
interval $(W_{j^{h-1}_{p'}-1}, W_{j^{h-1}_{p'+2}-1}]$
and their ancestors including $\rho_h$.
Let $\rho_{h+1}$ be the root of $T'_{copy}$.
Note that for each leaf of $T'_{copy}$, the original node in $T_h$ is a maximal subinterval node for $(W_{j^{h-1}_{p'}-1}, W_{j^{h-1}_{p'+2}-1}]$.
In $T'_{copy}$, we store function $\bar{\theta}^{\ell_u,+,j^{h-1}_{p'}}(c, z)$ 
at all the leaves
and {\sf NULL} at the other nodes.
Then, we connect nodes in $T'_{copy}$ with ones in $T_h$ as follows:
$\nu' \in T'_{copy}$ is connected with $\nu \in T_h$ if and only if
$\nu$ is not copied in $T'_{copy}$ and
the original node of $\nu'$ into $T_h$ is $p_{\nu}$.
See Fig.~\ref{fig:3}.

\begin{figure}[tb]
  \begin{center}
    \includegraphics[width=11cm,pagebox=cropbox,clip]{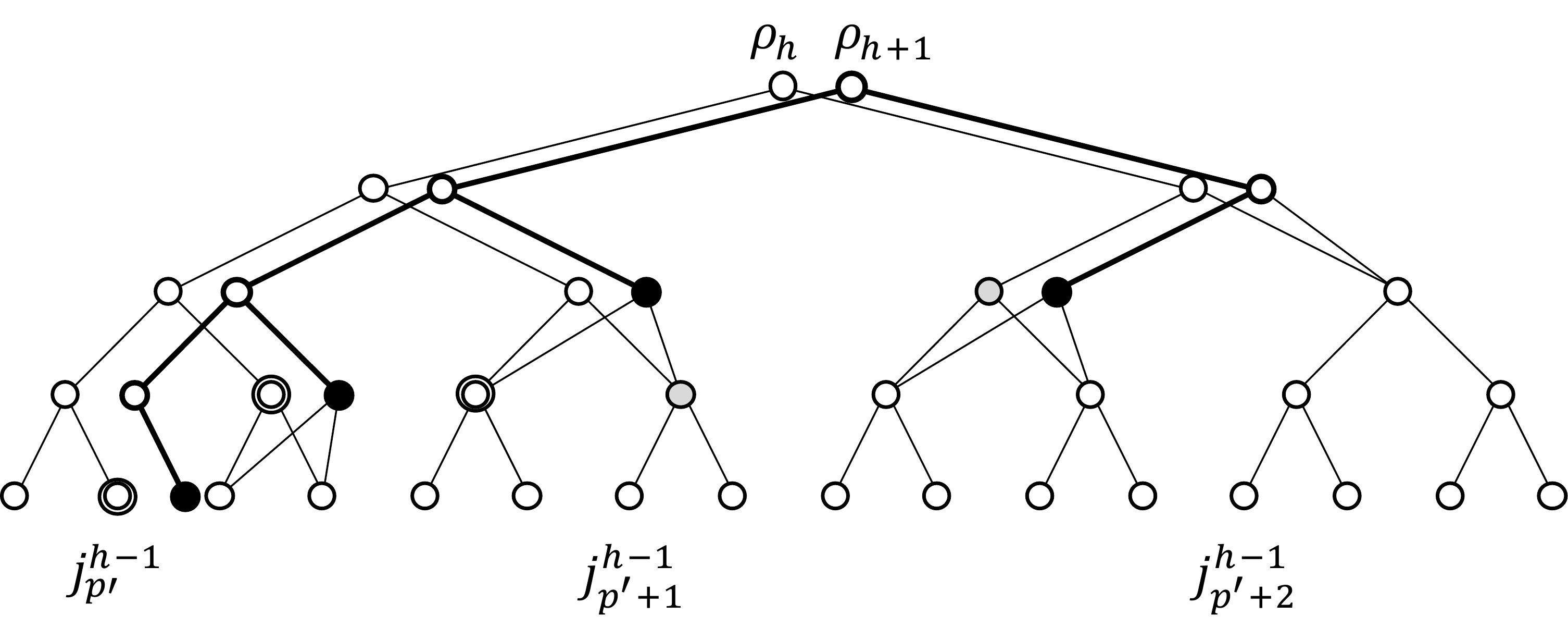}
    \caption{Illustration to show how to add the information of $T_{h+1}$ to $T_h$.
    The double circles indicate the maximal subinterval nodes for interval $(W_{j^{h-1}_{p'}-1}, W_{j^{h-1}_{p'+1}-1}]$.
    The gray circles indicate the maximal subinterval nodes for interval $(W_{j^{h-1}_{p'+1}-1}, W_{j^{h-1}_{p'+2}-1}]$.
    The thick subtree indicates $T'_{copy}$ and the black circle nodes store function $\bar{\theta}^{\ell_u,+,j^{h-1}_{p'}}(c, z)$.
    }
    \label{fig:3}
  \end{center}
\end{figure}

We note that the number of nodes in $T'$ (also $T'_{copy}$) is at most $O(\log m_u)$
because the number of maximal subinterval nodes for $(W_{j^{h-1}_{p'}-1}, W_{j^{h-1}_{p'+2}-1}]$ is at most $O(\log m_u)$ 
by the property of a segment tree 
and the number of these ancestors is also at most $O(\log m_u)$ in a binary tree. 
This implies that the above operation takes $O(\log m_u)$ time.
Repeating the above modification to $T_h$ for $h \in [1..m_u]$, 
we have a persistent segment tree containing all $T_h$ in $O(m_u \log m_u)$ time. 
Thus, we obtain the whole information of TYPE~{III} of ${\cal T}$ in time $O(n \log^2 n)$ by Property~\ref{property1}. 

Using this persistent segment tree, we can compute $\bar{\theta}^{\ell_u,+,[\ell_u+1.. r_u]}(c, z)$ for any $z \in [0,W_n]$ and any $c>0$ in $O(\log m_u)$ time as follows: 
Find integer $h$ over $[1..m_u]$ 
such that $c \in (c^{u,+}_{h},c^{u,+}_{h-1}]$ 
in $O(\log m_u)$ time by binary search, 
and then search in the persistent segment tree as $T_{h}$ in time $O(\log m_u)$. 

\begin{lemma}[Query with TYPE~III]\label{lem:type3}
Suppose that TYPE~{III} of ${\cal T}$ is available. 
Given a node $u \in {\cal T}$, real values $z \in [0, W_n]$ and $c>0$, we can obtain
 
\noindent 
 (i) index $H \in [\ell_u+1..r_u]$ such that 
$\bar{\theta}^{\ell_u,+,H}(c, z) = \bar{\theta}^{\ell_u, +,[\ell_u+1.. r_u]}(c, z)$, and

\noindent 
 (ii) index $H \in [\ell_u..r_u-1]$ such that 
 $\bar{\theta}^{r_u,-,H}(c, z) = \bar{\theta}^{r_u,-,[\ell_u.. r_u-1]}(c, z)$\\
 in time $O(\log n)$ respectively.
\end{lemma}


\noindent
{\bf TYPE~{IV}.} 
We give the information only for computing 
$\int^z_0 \bar{\theta}^{\ell_u,+,[\ell_u+1.. r_u]}(c, t)dt$ 
stored at $u \in {\cal T}$ as TYPE~{IV} 
since the case for $\int^{W_n}_z \bar{\theta}^{\ell_u,-,[\ell_u.. r_u-1]}(c, t)dt$ is symmetric.
Similar to TYPE~{III}, we prepare only for the case of $z \in (W_{\ell_u},W_{r_u}]$
since it holds that 
$\int^z_0 \bar{\theta}^{\ell_u,+,[\ell_u+1.. r_u]}(c, t)dt = 0$ for $z \in [0, W_{\ell_u}]$ and
\[
\int^z_0 \bar{\theta}^{\ell_u,+,[\ell_u+1.. r_u]}(c, t)dt = \int^{W_{r_u}}_0 \bar{\theta}^{\ell_u,+,[\ell_u+1.. r_u]}(c, t)dt + 
\frac{(z - W_{r_u})^2}{2c}
\]
for $z \in (W_{r_u}, W_n]$, of which the first term can be obtained by prepared information with 
$z = W_{r_u}$ and the second term by elementally calculation.

For each $u \in {\cal T}$, we construct a persistent segment tree again, 
which is similar to one shown in the previous section. 
To begin with, consider the case of $c \in (c^{u,+}_{h},c^{u,+}_{h-1}]$ 
with some $h \in [1..m_u]$ (where recall that $c^{u,+}_{0} = \infty$ and $c^{u,+}_{m_u} = 0$), 
and indices $j^{h-1}_1, \cdots, j^{h-1}_{m_u-h+1}$ 
that satisfy~{\eqref{eq:j_u_h}}. 
In this case, for $z \in (W_{j^{h-1}_p-1}, W_{j^{h-1}_{p+1}-1}]$ with $p \in [1..m_u-h+1]$,
we have
\begin{equation}\label{eq:bar_phifunction_p}
\int^z_0 \bar{\theta}^{\ell_u,+,[\ell_u+1.. r_u]}(c, t)dt
= \sum_{q=1}^{p-1} \left\{\int_{W_{j^{h-1}_q-1}}^{W_{j^{h-1}_{q+1}-1}} \bar{\theta}^{\ell_u,+,j^{h-1}_q}(c, t)dt \right\} + \int_{W_{j^{h-1}_p-1}}^{z} \bar{\theta}^{\ell_u,+,j^{h-1}_p}(c, t)dt.
\end{equation}
For ease of reference, we use $F^{h,p}(c,z)$ instead of the right hand side of 
\eqref{eq:bar_phifunction_p}. 

Similarly to the explanation for TYPE~III, 
let $T_h$ be a segment tree with root $\rho_h$ and 
$m_u$ leaves labeled as $l_u+1, \ldots, r_u$,
and each leaf $j$ of $T_h$ corresponds to interval $(W_{j-1}, W_{j}]$. 
In the same manner as for TYPE~{III},
for each $p \in [1..m_u-h+1]$, 
we store function $F^{h,p}(c,z)$ at all the maximal subinterval nodes in $T_h$ for interval $(W_{j^{h-1}_p-1}, W_{j^{h-1}_{p+1}-1}]$. 
Using $T_h$, for any $z \in (W_{\ell_u},W_{r_u}]$ and any $c \in (c^{u,+}_{h},c^{u,+}_{h-1}]$, 
we can compute value $\int^z_0 \bar{\theta}^{\ell_u,+,[\ell_u+1.. r_u]}(c, t)dt$ in time $O(\log m_u)$
by summing up all functions of nodes on a path 
from root $\rho_h$ to leaf with an interval that contains $z$. 
%
Actually, we store functions in a more complicated way in order to maintain them as a persistent data structure.
We construct a persistent segment tree at $u \in {\cal T}$ in $O(m_u\log m_u)$ time, 
in which we can search as if all of $T_h$ are maintained, as follows.

First of all, we construct a segment tree $T_{1}$ for $c \in (c^{u,+}_{1},\infty]$ 
such that all inner nodes is stored constant function 0, 
and each leaf $j \in [\ell_u + 1..r_u]$ has function
\[
F^{1,j-\ell_u}(c,z)
= \sum_{i=\ell_u+1}^{j-1} \left\{\int_{W_{i-1}}^{W_{i}} \bar{\theta}^{\ell_u,+,i}(c, t)dt \right\} + \int_{W_{j-1}}^{z} \bar{\theta}^{\ell_u,+,j}(c, t)dt.
\]
in $O(m_u)$ time.

We update this tree as follows: 
Suppose that a persistent segment tree contains the information of $T_h$. 
We update this persistent segment tree by adding the information of $T_{h+1}$. 
Let $p'$ be the maximizer in \eqref{eq:c_u+_h},
i.e., when $c = c^{u,+}_h$, 
$\bar{\theta}^{\ell_u,+,j^{h-1}_{p'}}(c, z)$ and $\bar{\theta}^{\ell_u,+,j^{h-1}_{p'+1}}(c, z)$ come to overlap each other on $(W_{j^{h-1}_{p'+1}-1},W_{j^{h-1}_{p'+2}-1}]$.
Therefore, for $z \in (W_{j^{h-1}_{p'+1}-1}, W_{j^{h-1}_{p'+2}-1}]$, 
\begin{eqnarray*}
F^{h+1,p'+1}(c,z) - F^{h,p'+1}(c,z) =
\int_{W_{j^{h-1}_{p'+1}-1}}^{z} \left\{
- \bar{\theta}^{\ell_u,+,j^{h-1}_{p'+1}}(c, t) + 
\bar{\theta}^{\ell_u,+,j^{h-1}_{p'}}(c, t) \right\}dt,
\end{eqnarray*}
and for $z \in (W_{j^{h-1}_{q}-1}, W_{j^{h-1}_{q+1}-1}]$ with $q \in [p'+2..m_u-h+1]$, 
\begin{equation*}
F^{h+1,q}(c,z) - F^{h,q}(c,z) =
\int_{W_{j^{h-1}_{p'+1}-1}}^{W_{j^{h-1}_{p'+2}-1}}
\left\{ - \bar{\theta}^{\ell_u,+,j^{h-1}_{p'+1}}(c,t) + \bar{\theta}^{\ell_u,+,j^{h-1}_{p'}}(c,t) \right\} dt.
\end{equation*}
This implies that we obtain $T_{h+1}$ 
by adding
$\int_{W_{j^{h-1}_{p'+1}-1}}^{z} \left\{- \bar{\theta}^{\ell_u,+,j^{h-1}_{p'+1}}(c, t) + 
\bar{\theta}^{\ell_u,+,j^{h-1}_{p'}}(c, t) \right\}dt$
and  
$\int_{W_{j^{h-1}_{p'+1}-1}}^{W_{j^{h-1}_{p'+2}-1}}
\left\{ - \bar{\theta}^{\ell_u,+,j^{h-1}_{p'+1}}(c,t) + \bar{\theta}^{\ell_u,+,j^{h-1}_{p'}}(c,t) \right\} dt$
at all the maximal subinterval nodes in $T_h$ for intervals 
$(W_{j^{h-1}_{p'+1}-1}, W_{j^{h-1}_{p'+2}-1}]$ and $(W_{j^{h-1}_{p'+2}-1}, W_{r_u}]$, respectively.

Similar to TYPE~{III}, 
we prepare a copy of subtree $T'$ of $T_{h}$, say $T'_{copy}$,
where $T'$ is the minimal subtree containing all the maximal subinterval nodes in $T_{h}$ for interval $(W_{j^{h-1}_{p'+1}-1}, W_{r_u}]$
and their ancestors including $\rho_h$.
Let $\rho_{h+1}$ be the root of $T'_{copy}$.
Note that for each leaf of $T'_{copy}$, the original node in $T_h$ is a maximal subinterval node for $(W_{j^{h-1}_{p'+1}-1}, W_{j^{h-1}_{p'+2}-1}]$ or $(W_{j^{h-1}_{p'+2}-1}, W_{r_u}]$.
In $T'_{copy}$, we add $\int_{W_{j^{h-1}_{p'+1}-1}}^{z} \left\{- \bar{\theta}^{\ell_u,+,j^{h-1}_{p'+1}}(c, t) + 
\bar{\theta}^{\ell_u,+,j^{h-1}_{p'}}(c, t) \right\}dt$
at the leaves corresponding to $(W_{j^{h-1}_{p'+1}-1}, W_{j^{h-1}_{p'+2}-1}]$,
and  
$\int_{W_{j^{h-1}_{p'+1}-1}}^{W_{j^{h-1}_{p'+2}-1}}
\left\{ - \bar{\theta}^{\ell_u,+,j^{h-1}_{p'+1}}(c,t) + \bar{\theta}^{\ell_u,+,j^{h-1}_{p'}}(c,t) \right\} dt$
at the leaves corresponding to $(W_{j^{h-1}_{p'+2}-1}, W_{r_u}]$.
Then, we connect nodes in $T'_{copy}$ with ones in $T_h$ as follows:
$\nu' \in T'_{copy}$ is connected with $\nu \in T_h$ if and only if
$\nu$ is not copied into $T'_{copy}$ and
the original node of $\nu'$ in $T_h$ is $p_{\nu}$.
See Fig.~\ref{fig:4}.
\begin{figure}[t]
  \begin{center}
    \includegraphics[width=11cm,pagebox=cropbox,clip]{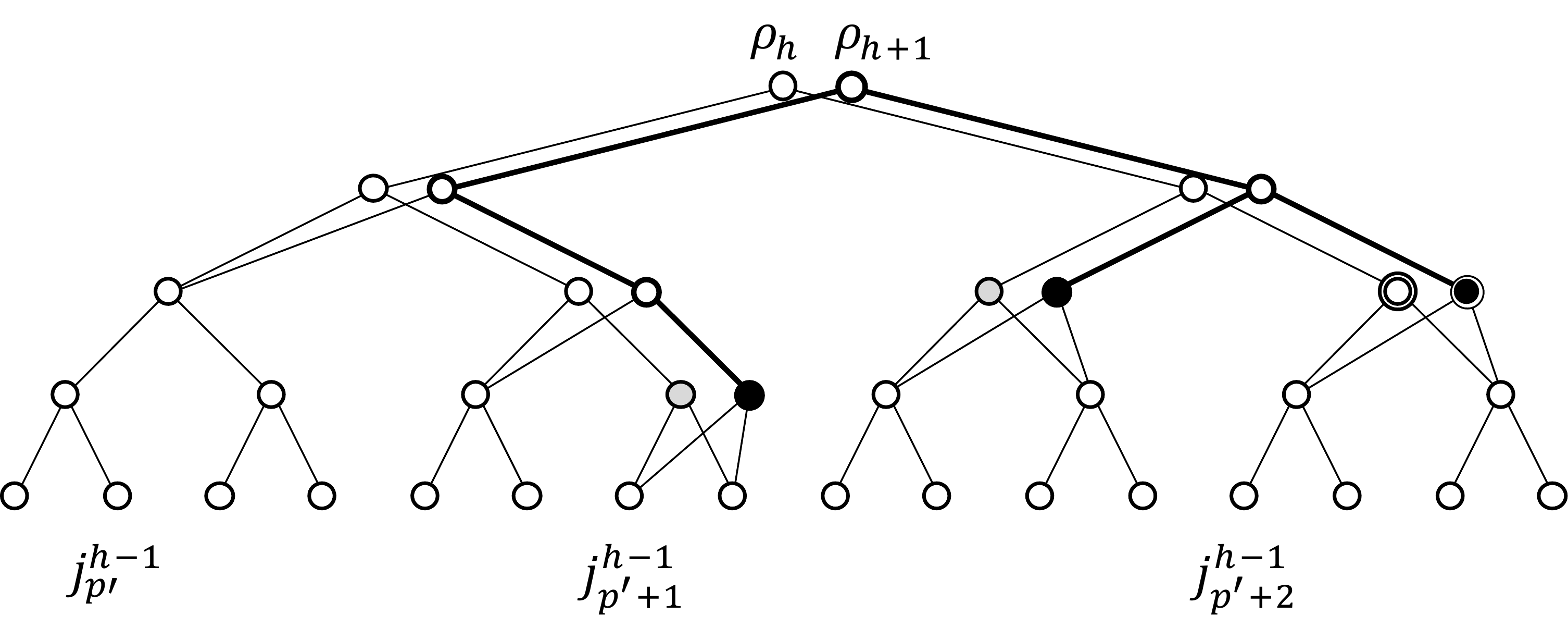}
    \caption{Illustration to show how to add the information of $T_{h+1}$ to $T_h$.
    The gray circles indicate the maximal subinterval nodes for interval $(W_{j^{h-1}_{p'+1}-1}, W_{j^{h-1}_{p'+2}-1}]$.
    The double circles indicate the maximal subinterval nodes for interval $(W_{j^{h-1}_{p'+2}-1}, W_{r_u}]$. 
    The thick subtree indicates $T'_{copy}$,
    the black circle nodes store functions obtained by adding 
    $\int_{W_{j^{h-1}_{p'+1}-1}}^{z} \left\{- \bar{\theta}^{\ell_u,+,j^{h-1}_{p'+1}}(c, t) + \bar{\theta}^{\ell_u,+,j^{h-1}_{p'}}(c, t) \right\}dt$ 
    to functions stored in the gray circle nodes, and
    the black double circle nodes store functions obtained by adding $\int_{W_{j^{h-1}_{p'+1}-1}}^{W_{j^{h-1}_{p'+2}-1}}
    \left\{ - \bar{\theta}^{\ell_u,+,j^{h-1}_{p'+1}}(c,t) + \bar{\theta}^{\ell_u,+,j^{h-1}_{p'}}(c,t) \right\} dt$ 
    to functions stored in the double circle nodes.}
    \label{fig:4}
  \end{center}
\end{figure}
Repeating the above modification for $h \in [1..m_u]$, 
we have a persistent segment tree containing all $T_h$ in $O(m_u \log m_u)$ time.

Using this persistent segment tree, we can compute 
$\int_{0}^{z} \bar{\theta}^{\ell_u, +, [\ell_u+1.. r_u]}(c,t)dt$
for any $z \in [0,W_n]$ and any $c >0$ in $O(\log m_u)$ time in the same manner as for TYPE~III.

\begin{lemma}[Query with TYPE~{IV}]\label{lem:type4}
Suppose that TYPE~{IV} of ${\cal T}$ is available. 
Given a node $u \in {\cal T}$, real values $z \in [0, W_n]$ and $c>0$, we can obtain
(i) value $\int_{0}^{z} \bar{\theta}^{\ell_u, +, [\ell_u+1.. r_u]}(c,t)dt$, and (ii) value $\int_{z}^{W_n} \bar{\theta}^{r_u, -, [\ell_u.. r_u-1]}(c,t)dt$ 
in time $O(\log n)$ respectively.
\end{lemma}

\section{Query Time for $\opt(i,j)$: Proof of Lemma~\ref{lem:general_query}}\label{app:gen_edge_capacity}

In this section, we give an algorithm to prove Lemma~\ref{lem:general_query} for the non-confluent flow model
since a similar argument holds even for the confluent flow model with minor modification. 
Assume that the data structure ${\cal T}$ with TYPEs~{I, II, III and IV}
introduced in Section~\ref{sec:data_structure} is available. 
We recall the outline for computing $\opt(i,j)$ given two integers $i, j \in [1..n]$ with $i < j$.

\noindent
Phase~{1}: Compute all the maximal subpath nodes for $P_{i+1,j-1}$ by walking on ${\cal T}$ from leaf with a label $i+1$ to one with $j-1$. 
Let $U = \{u_1, u_2, \ldots, u_m\}$ be a set of such the maximal subpath nodes for $P_{i,j}$
so that $r_{u_s}+1=l_{u_{s+1}}$ for $s \in [0..m]$ 
where we define $r_{u_0}=i$ and $l_{u_{m+1}}=j$. 

\noindent
Phase~{2}: 
For each $u \in U$, 
compute a real interval ${\cal I}^+_u$ such that 
$\theta^{i,+}(z) = \theta^{i,+,[\ell_u..r_u]}(z)$ holds on any $z \in {\cal I}^+_u$, 
and a real interval ${\cal I}^-_u$ such that 
$\theta^{j,-}(z) = \theta^{j,-,[\ell_u..r_u]}(z)$ holds on any $z \in {\cal I}^-_u$,
both of which are obtained by 
using TYPEs I and III stored at $u \in {\cal T}$.

\noindent
Phase~{3}: 
Compute the pseudo-intersection point $z^*$ of $\theta^{i,+}(z)$ and $\theta^{j,-}(z)$ on $[W_{i}, W_{j-1}]$ 
by using real intervals obtained Phase~{2}. 

\noindent
Phase~{4}: Compute $\opt(i,j) = \Phi^{i,j}(z^*)$ as follows:
By formula~\eqref{eq:at_subpath}, we have 
\begin{eqnarray*}
\Phi^{i,j}(z^*) &=& \int_{0}^{z^*} \theta^{i, +}(t)dt + \int_{z^*}^{W_n} \theta^{j, -}(t)dt \\
&=& 
\sum_{u \in U}\left\{
\int_{{\cal I}^+_u \cap [0, z^*]} \theta^{i, +, [\ell_u..r_u]}(t)dt + 
\int_{{\cal I}^-_u \cap [z^*, W_n]} \theta^{j, -, [\ell_u..r_u]}(t)dt
\right\}. 
\end{eqnarray*}
Compute each integral $\int\theta^{i, +, [\ell_u..r_u]}(t)dt$ and 
$\int\theta^{j, -, [\ell_u..r_u]}(t)dt$ 
for each $u \in U$ 
using TYPEs II and IV stored at $u$.


In the following, we see the details of Phases~{2},~{3}, and~{4}. 

\subsection{Phase~{2}}\label{app:phase2}

In Phase~{2}, we construct, for $u \in U$, 
real intervals ${\cal I}^+_u$ and ${\cal I}^-_u$ such that
$\theta^{i,+}(z) = \theta^{i,+,[\ell_u..r_u]}(z)$ holds on any $z \in {\cal I}^+_u$
and 
$\theta^{j,-}(z) = \theta^{j,-,[\ell_u..r_u]}(z)$ holds on any $z \in {\cal I}^-_u$,
respectively. 
We give only the computation for ${\cal I}^+_u$ 
since ${\cal I}^-_u$ can be constructed in a symmetric manner. 
First confirm by formula~\eqref{eq:function_right} and Definition~\ref{def:subup} that 
$\theta^{i,+}(z) = \theta^{i,+,[i+1..h]}(z)$ for $z \in [W_i,W_h]$ with $h \ge i+1$. 
Thus, it is enough to obtain the information of $\theta^{i,+,[i+1..j-1]}(z)$, 
since we consider the case with $z \in [W_i, W_{j-1}]$. 

For each $u \in U$, ${\cal I}^+_u$ consists of 
three consecutive intervals ${\cal I}^{+}_{u,1}$, ${\cal I}^{+}_{u,2}$ and ${\cal I}^{+}_{u,3}$ which satisfy the following conditions:
\begin{eqnarray}\label{eq:interval_condition_phase2}
\theta^{i,+}(z) = 
\left\{
\begin{array}{ll}
\theta^{i,+,\ell_u}(z) & \text{ if } z \in {\cal I}^{+}_{u,1}, \\
\bar{\theta}^{\ell_u,+,[\ell_u+1..r_u]}(C_{i,\ell_u}, z)+\tau \cdot L_{i,\ell_u}
& \text{ if } z \in {\cal I}^{+}_{u,2},  \\
\theta^{\ell_u,+,[\ell_u+1..r_u]}(z)+\tau \cdot L_{i,\ell_u} 
& \text{ if } z \in {\cal I}^{+}_{u,3}.
\end{array}
\right.
\end{eqnarray}
Moreover, these intervals have the following forms:
\begin{eqnarray*}
{\cal I}^+_{u_s,h} = 
\left\{
\begin{array}{ll}
[W_{i}, \beta_{1,1}] & \quad \text{if } s=1 \text{ and } h = 1,\\
(\alpha_{s,h}, \beta_{s,h}] & \quad \text{otherwise }
\end{array}
\right.
\end{eqnarray*}
such that $\alpha_{s+1,1} = \beta_{s,3}$ and $\alpha_{s,h+1} = \beta_{s,h}$ for $h = 1,2$.
Note that, for any value $\alpha$, 
an interval $(\alpha, \alpha]$ denotes the empty interval. 


In Phase~2, we inductively construct ${\cal I}^{+}_{u_s, 1}$, ${\cal I}^{+}_{u_s, 2}$ and ${\cal I}^{+}_{u_s, 3}$ for $s \in [1..m]$ as follows:

\noindent
{\bf [Induction hypothesis]}
Assume that it has been obtained 
${\cal I}^{+}_{u_t, h}$ for all $t \in [1..s]$ and $h \in [1..3]$
such that
\[
\theta^{i, +, [i+1..r_{u_s}]}(z) = \theta^{i, +, [\ell_{u_t}..r_{u_t}]}(z)
\]
for any $z \in {\cal I}^+_{u_t}$. 

\noindent
{\bf [Induction step]}
The induction step consists of three substeps: 
\begin{description}
    \item[Substep~1.] 
    Compute an interval ${\cal J}_1$ such that 
    \[
    \theta^{i, +, [i+1..\ell_{u_{s+1}}]}(z) = \theta^{i, +, \ell_{u_{s+1}}}(z)
    \] 
    for any $z \in {\cal J}_1$, and 
    update ${\cal I}^{+}_{u_t, h}$ for all $t \in [1..s]$ and $h \in [1..3]$
    so that
    \[
    \theta^{i, +, [i+1..\ell_{u_{s+1}}]}(z) = \theta^{i, +, [\ell_{u_t}..r_{u_t}]}(z)
    \]
    for any $z \in {\cal I}^+_{u_t}$. 
    \item[Substep~2.]
    Compute intervals ${\cal J}_2$ and ${\cal J}_3$ such that 
    \begin{eqnarray}
    \theta^{i, +, [\ell_{u_{s+1}}+1..r_{u_{s+1}}]}(z) = \left\{
    \begin{array}{l}
    \bar{\theta}^{\ell_{u_{s+1}},+,[\ell_{u_{s+1}}+1..r_{u_{s+1}}]}(C_{i,\ell_{u_{s+1}}},z)+\tau \cdot L_{i,\ell_{u_{s+1}}} \\
    \qquad\qquad\qquad\qquad\qquad \text{ for any }z \in {\cal J}_2,\\
    \theta^{\ell_{u_{s+1}},+,[\ell_{u_{s+1}}+1..r_{u_{s+1}}]}(z)+\tau \cdot L_{i,\ell_{u_{s+1}}} \\ \qquad\qquad\qquad\qquad\qquad \text{ for any } z \in {\cal J}_3.
    \end{array}
    \right. \notag
    \end{eqnarray}
    \item[Substep~3.]
    Based on intervals ${\cal J}_1$, ${\cal J}_2$ and ${\cal J}_3$, compute ${\cal I}^{+}_{u_{s+1}, h}$ for all $h \in [1..3]$
    such that $\theta^{i, +, [i+1..r_{u_{s+1}}]}(z) = \theta^{i, +, [\ell_{u_{s+1}}..r_{u_{s+1}}]}(z)$
    for any $z \in {\cal I}^+_{u_{s+1}}$, and
    update ${\cal I}^{+}_{u_t, h}$ for all $t \in [1..s]$ and $h \in [1..3]$
    so that
    $\theta^{i, +, [i+1..r_{u_{s+1}}]}(z) = \theta^{i, +, [\ell_{u_t}..r_{u_t}]}(z)$
    for any $z \in {\cal I}^+_{u_t}$. 
\end{description}

Let us show the details of each of three substeps of the induction step. 





\noindent
{\bf Substep 1.} In this step, 
we compute an interval ${\cal J}_1 = (\gamma, W_{j-1}]$
such that $\theta^{i, +, [i+1..r_{u_{s}}]}(z) \leq \theta^{i, +, \ell_{u_{s+1}}}(z)$ for any $z \in {\cal J}_1$. 
It means that $\theta^{i, +, [i+1..\ell_{u_{s+1}}]}(z) = \theta^{i, +, \ell_{u_{s+1}}}(z)$ for any $z \in {\cal J}_1$.
And then we update $s$ intervals 
${\cal I}^+_{u_t}$ for $t \in [1..s]$
such that $\theta^{i, +, [i+1..\ell_{u_{s+1}}]}(z) = \theta^{i, +, [\ell_{u_t}..r_{u_t}]}(z)$
for any $t \in [1..s]$ and any $z \in {\cal I}^+_{u_t}$. 

Every slope of line segments of $\theta^{i, +, [i+1..r_{u_{s}}]}(z)$
is at most $1/C_{i,r_{u_s}}$ and the slope of 
$\theta^{i, +, \ell_u}(z)$ is $1/C_{i,\ell_{u_{s+1}}}$.
Since $C_{i,\ell_{u_{s+1}}} = \min\{c_{r_{u_s}}, C_{i,r_{u_s}}\} \leq C_{i,r_{u_s}}$, we have $1/C_{i,r_{u_s}} \leq 1/C_{i,\ell_{u_{s+1}}}$ and 
it implies that 
$\theta^{i, +, \ell_{u_{s+1}}}(z) - \theta^{i, +, [i+1..r_{u_{s}}]}(z)$
is non-decreasing in $z \in (W_{\ell_{u_{s+1}}}, W_n]$.
Thus, the value $\gamma$ is the maximum value $z \in [W_{\ell_{u_{s+1}}}, W_{j-1}]$ satisfying
\begin{eqnarray}\label{eq:condition_phase2_1}
\theta^{i, +, [i+1..r_{u_{s}}]}(z) - \theta^{i, +, \ell_{u_{s+1}}}(z) < 0.
\end{eqnarray}
We find such $\gamma$ in $O(\log^2 n)$ time as follows.

\noindent
[Substep 1-1.] Find the maximum value $\beta \in \left\{\beta_{t,h} \mid t \in [1..s], h \in [1..3]\right\}$ 
satisfying 
\begin{eqnarray}
\theta^{i, +, [i+1, r_{u_s}]}(\beta) - 
\theta^{i, +, \ell_{u_{s+1}}}(\beta)
\geq 0.
\end{eqnarray}
We calculate the second term $\theta^{i, +, \ell_{u_{s+1}}}(\beta)$
by elementally calculation of formula~\eqref{eq:subfunction_right} in $O(1)$ time. 
For computing the first term $\theta^{i, +, [i+1, r_{u_s}]}(\beta)$, 
we note that 
\[
\theta^{i, +, [i+1, r_{u_t}]}(\beta_{t,h}) = \theta^{i, +, [\ell_{u_t}, r_{u_t}]}(\beta_{t,h})
\]
holds because
$\theta^{i, +, [i+1, r_{u_t}]}(z) = \theta^{i, +, [\ell_{u_t}, r_{u_t}]}(z)$
for any $z \in {\cal I}^{+}_{u_t}$. 
Let us consider three cases, depending on an integer $h \in [1..3]$ of $\beta = \beta_{t,h}$: 
\begin{enumerate}[\textrm{Case} (i)]
\item If $\beta = \beta_{t,1}$, we have
\[
\theta^{i, +, [\ell_{u_t}, r_{u_t}]}(\beta_{t,1}) = 
\theta^{i, +, \ell_{u_t}}(\beta_{t,1}),
\]
which is obtained $O(1)$ time by elementally calculation of formula~\eqref{eq:subfunction_right}. 

\item If $\beta = \beta_{t,2}$, we have
\[
\theta^{i, +, [\ell_{u_t}, r_{u_t}]}(\beta_{t,2}) = 
\bar{\theta}^{\ell_u, +, [\ell_{u_t}, r_{u_t}]}(C_{i, \ell_{u_t}}, \beta_{t,2}) + \tau \cdot L_{i, \ell_u}, 
\]
in which $\bar{\theta}^{\ell_u, +, [\ell_{u_t}, r_{u_t}]}(C_{i, \ell_{u_t}}, \beta_{t,2})$ can be obtained in $O(\log n)$ time by Lemma~\ref{lem:type3}.

\item 
If $\beta = \beta_{t,3}$, we have
\[
\theta^{i, +, [\ell_{u_t}, r_{u_t}]}(\beta_{t,3}) = 
\theta^{\ell_u, +, [\ell_{u_t}, r_{u_t}]}(\beta_{t,3}) + \tau \cdot L_{i, \ell_u},
\]
in which $\theta^{\ell_u, +, [\ell_{u_t}, r_{u_t}]}(C_{i, \ell_{u_t}}, \beta_{t,3})$ can be obtained in $O(\log n)$ time by Lemma~\ref{lem:type1}. 
Therefore, this operation requires $O(\log^2 n)$ time since we have $s \leq m = O(\log n)$.
If we obtain $\beta = \beta_{s,3} = W_{j-1}$, 
we set $\gamma = W_{j-1}$, that is, ${\cal J}_1$ is the empty interval. 
It implies that all of ${\cal I}^{+}_{u_{s+1}, 1}$, ${\cal I}^{+}_{u_{s+1}, 2}$, and ${\cal I}^{+}_{u_{s+1}, 3}$ are the empty intervals and we are done.
Otherwise, we have two integers $t$ and $h$ such that 
$\gamma \in {\cal I}^{+}_{u_t, h}$ by using $\beta$. 
\end{enumerate}

\noindent
[Substep 1-2.] 
Letting us suppose that $\gamma \in {\cal I}^{+}_{u_t, h} = (\alpha_{u,h}, \beta_{u,h}]$, 
we can rewrite the condition~\eqref{eq:condition_phase2_1} as follows:
The value $\gamma$ is the maximum value $z \in {\cal I}^{+}_{u_t, h}$ satisfying
\begin{eqnarray}\label{eq:condition_phase2_2}
\theta^{i, +, [\ell_{u_t}..r_{u_t}]}(z) - \theta^{i, +, \ell_{u_{s+1}}}(z) < 0.
\end{eqnarray}
We find such $z$ as follows: 

\begin{enumerate}[\textrm{Case} (i)]
\item 
If $\gamma \in {\cal I}^{+}_{u_t, 1}$, then we have 
\[
\theta^{i, +, [\ell_{u_t}..r_{u_t}]}(z) = \theta^{i, +, \ell_{u_t}}(z),
\]
which is a linear function on ${\cal I}^{+}_{u_t, h}$.
Thus, the value $\gamma$ is the intersection point of 
$\theta^{i, +, \ell_{u_t}}(z)$ and $\theta^{i, +, \ell_{u_{s+1}}}(z)$,
which is obtained in $O(1)$ time.

\item If $\gamma \in {\cal I}^{+}_{u_t, 2}$, then we have 
\[
\theta^{i, +, [\ell_{u_t}..r_{u_t}]}(z) = \bar{\theta}^{i, +, [\ell_{u_t}+1..r_{u_t}]}(C_{i, \ell_{u_t}}, z)
= \bar{\theta}^{\ell_{u_t}, +, [\ell_{u_t}+1..r_{u_t}]}(C_{i, \ell_{u_t}}, z) + \tau \cdot L_{i, \ell_{u_t}}.
\]
Since functions $\bar{\theta}^{i, +, [\ell_{u_t}+1..r_{u_t}]}(C_{i, \ell_{u_t}}, z)$ and 
$\bar{\theta}^{\ell_{u_t}, +, [\ell_{u_t}+1..r_{u_t}]}(C_{i, \ell_{u_t}}, z)$
consists of line segments with the same slope $1 / C_{i, \ell_{u_t}}$, 
all breakpoints of these functions on ${\cal I}^{+}_{u_t,2}$ are 
contained in $\{\alpha_{t,2}, \beta_{t,2}\} \cup \{W_h \mid h \in [1..n] \text{ and } W_h \in {\cal I}^{+}_{u_t,2}\}$.

We find the minimum value $\gamma' \in \{\alpha_{t,2}, \beta_{t,2}\} \cup \{W_h \mid h \in [1..n] \text{ and } W_h \in {\cal I}^{+}_{u_t,2}\}$ such that 
\begin{eqnarray*}
& &\theta^{i, +, [\ell_{u_t}..r_{u_t}]}(z) - \theta^{i, +, \ell_{u_{s+1}}}(z) \\
& = & \bar{\theta}^{\ell_{u_t}, +, [\ell_{u_t}+1..r_{u_t}]}(C_{i, \ell_{u_t}}, z) + \tau \cdot L_{i, \ell_{u_t}}
- \theta^{i, +, \ell_{u_{s+1}}}(z) \leq 0
\end{eqnarray*}
by binary search.
Because we can obtain the term 
$\bar{\theta}^{\ell_{u_t}, +, [\ell_{u_t}+1..r_{u_t}]}(C_{i, \ell_{u_t}}, z)$
for any $z$ in $O(\log n)$ time by Lemma~\ref{lem:type3}, 
this binary search takes $O(\log^2 n)$ time in total. 
Using the value $\gamma'$, we obtain an integer $H$ such that 
$\bar{\theta}^{\ell_{u_t}, +, [\ell_{u_t}+1..r_{u_t}]}(C_{i, \ell_{u_t}}, \gamma')
= \bar{\theta}^{\ell_{u_t}, +, H}(C_{i, \ell_{u_t}}, \gamma')$ holds 
by Lemma~\ref{lem:type3} in $O(\log n)$ time. 
It implies that the value $\gamma$ is the intersection point of two line segments 
$\bar{\theta}^{\ell_{u_t}, +, H}(C_{i, \ell_{u_t}}, z) + \tau \cdot L_{i,\ell_{u_t}}$ and $\theta^{i, +, \ell_{u_{s+1}}}(z)$, which is computed by elementally calculation. 

\item If $\gamma \in {\cal I}^{+}_{u_t, 3}$, then we have 
\[
\theta^{i, +, [\ell_{u_t}..r_{u_t}]}(z) = 
\theta^{\ell_{u_t}, +, [\ell_{u_t}+1..r_{u_t}]}(z) + \tau \cdot L_{i, \ell_{u_t}}.
\]
Recall ${\cal B}^{u_{t}, +} = (b^{u_{t}, +}_{1}, b^{u_{t}, +}_{2}, ...)$ is breakpoints of $\theta^{\ell_{u_t}, +, [\ell_{u_t}+1..r_{u_t}]}(z)$. 
We find the minimum value $\gamma' \in \{\alpha_{t,3}, \beta_{t,3}\} \cup \{b^{u_{t}, +}_{p} \mid b^{u_{t}, +}_{p} \in {\cal I}^{+}_{u_t,2}\}$ such that 
\begin{eqnarray*}
\theta^{i, +, [\ell_{u_t}..r_{u_t}]}(z) - \theta^{i, +, \ell_{u_{s+1}}}(z) 
 =  
\theta^{\ell_{u_t}, +, [\ell_{u_t}+1..r_{u_t}]}(z) + \tau \cdot L_{i, \ell_{u_t}}
- \theta^{i, +, \ell_{u_{s+1}}}(z) \leq 0
\end{eqnarray*}
by binary search, which requires $O(\log n)$ time
because the term $\theta^{\ell_{u_t}, +, [\ell_{u_t}+1..r_{u_t}]}(z)$ 
is computed for $z = \alpha_{t,3}, \beta_{t,3}$
in $O(\log n)$ time by Lemma~\ref{lem:type1}
and for $z = b^{u_t,+}_{p}$ in $O(1)$ time by using the list of TYPE~I. 
Then, we obtain an integer $H$ such that 
$\theta^{\ell_{u_t}, +, [\ell_{u_t}+1..r_{u_t}]}(\gamma')
=\theta^{\ell_{u_t}, +, H}(C_{i, \ell_{u_t}}, \gamma')$ holds 
by Lemma~\ref{lem:type1} in $O(\log n)$ time. 
It implies that the value $\gamma$ is the intersection point of two line segments 
$\theta^{\ell_{u_t}, +, H}(z) + \tau \cdot L_{i,\ell_{u_t}}$ and $\theta^{i, +, \ell_{u_{s+1}}}(z)$, which is computed by elementally calculation. 
\end{enumerate}

Thus, we can obtain $\gamma$ in $O(\log^2 n)$ time
and done the computation of interval ${\cal J}$. 
In order to complete the first step, 
we update intervals 
that have an overlap with ${\cal J}$ as follows:
In the above step, if $\gamma \in {\cal I}^{+}_{u_t, h} = (\alpha_{t,h}, \beta_{t,h}]$, 
then we set ${\cal I}^{+}_{u_t, h} = (\alpha_{t,h}, \gamma]$. 
Moreover, we set
all intervals ${\cal I}^{+}_{u_t, h'}$ for $h'>h$ and ${\cal I}^{+}_{u_t', h'}$ for $t'>t$ and $h' \in [1..3]$
to the empty interval $(\gamma, \gamma]$.\\ 

\noindent
{\bf Substep 2.} 
In the second step, first of all, 
we find the smallest integer $q$ such that $C_{i, q} < C_{i, \ell_{u_{s+1}}}$ in $O(\log n)$ time
by binary search.
Using this $q$, we have
\begin{eqnarray}\label{eq:phase2_3}
& &\theta^{i,+,[\ell_{u_{s+1}}+1..r_{u_{s+1}}]}(z) \notag \\
&= &
\max\{
\bar{\theta}^{i,+,[\ell_{u_{s+1}}+1..q-1]}(C_{i, \ell_{u_{s+1}}}, z), 
\theta^{i,+,[\ell_{u_{s+1}}+1..r_{u_{s+1}}]}(z)\} \notag \\
&= &
\max\{
\bar{\theta}^{\ell_{u_{s+1}},+,[\ell_{u_{s+1}}+1..q-1]}(C_{i, \ell_{u_{s+1}}}, z), 
\theta^{\ell_{u_{s+1}},+,[\ell_{u_{s+1}}+1..r_{u_{s+1}}]}(z)\} + \tau \cdot L_{i, \ell_{u_{s+1}}}. 
\end{eqnarray}
Note that 
$\bar{\theta}^{\ell_{u_{s+1}},+,[\ell_{u_{s+1}}+1..q-1]}(C_{i, \ell_{u_{s+1}}}, z)
- \theta^{\ell_{u_{s+1}},+,[\ell_{u_{s+1}}+1..r_{u_{s+1}}]}(z)$
is non-decreasing on $(W_{q-1}, W_{j-1}]$.
Thus, there exists a value $\delta \in (W_{q-1}, W_{j-1}]$ such that 
\[
\bar{\theta}^{\ell_{u_{s+1}},+,[\ell_{u_{s+1}}+1..q-1]}(C_{i, \ell_{u_{s+1}}}, \delta) - 
\theta^{\ell_{u_{s+1}},+,[\ell_{u_{s+1}}+1..r_{u_{s+1}}]}(\delta)
\leq 0 
\]
and 
\[
\bar{\theta}^{\ell_{u_{s+1}},+,[\ell_{u_{s+1}}+1..q-1]}(C_{i, \ell_{u_{s+1}}}, \delta - \epsilon) -
\theta^{\ell_{u_{s+1}},+,[\ell_{u_{s+1}}+1..r_{u_{s+1}}]}(\delta - \epsilon) 
> 0 
\]
hold for any sufficiently small positive $\epsilon$.
We find such the value $\delta$ and set two intervals 
${\cal J}_2 = (W_{\ell_{u_{s+1}}}, \delta]$ and ${\cal J}_3 = (\delta, W_{j-1}]$.
Equation~\eqref{eq:phase2_3} implies that we have
\begin{eqnarray}\label{eq:phase2_4}
\theta^{i,+,[\ell_{u_{s+1}}+1..r_{u_{s+1}}]}(z)
= 
\left\{
\begin{array}{ll}
\bar{\theta}^{\ell_{u_{s+1}},+,[\ell_{u_{s+1}}+1..q-1]}(C_{i, \ell_{u_{s+1}}}, z) + \tau \cdot L_{i, \ell_{u_{s+1}}} &  \text{ if } z \in {\cal J}_2, \\
\theta^{\ell_{u_{s+1}},+,[\ell_{u_{s+1}}+1..r_{u_{s+1}}]}(z)   + \tau \cdot L_{i, \ell_{u_{s+1}}}  &  \text{ if } z \in {\cal J}_3.
\end{array}
\right.
\end{eqnarray}

In order to obtain the value $\delta$, 
we find integers $h_2$ such that 
\[
\bar{\theta}^{\ell_{u_{s+1}},+,[\ell_{u_{s+1}}+1..q-1]}(C_{i, \ell_{u_{s+1}}}, W_{q-1}) = 
\bar{\theta}^{\ell_{u_{s+1}},+,h_2}(C_{i, \ell_{u_{s+1}}}, W_{q-1}) 
\]
holds in $O(\log n)$ time by Lemma~\ref{lem:type3}. 
Note that we have 
\[
\bar{\theta}^{\ell_{u_{s+1}},+,[\ell_{u_{s+1}}+1..q-1]}(C_{i, \ell_{u_{s+1}}}, z) = 
\bar{\theta}^{\ell_{u_{s+1}},+,h_2}(C_{i, \ell_{u_{s+1}}}, z)
\]
for any $z \in (W_{q-1}, W_{j-1}]$. 
Next, we find the minimum value $\delta'$ over the breakpoints $\{b^{u_{s+1}, +}_{p}
\}$ of $\theta^{\ell_{u_{s+1}}, +, [\ell_{u_{s+1}}+1..r_{u_{s+1}}]}(z)$
such that 
\begin{eqnarray*}
\theta^{\ell_{u_{s+1}}, +, h_2}(z) - \theta^{\ell_{u_{s+1}}, +, [\ell_{u_{s+1}}+1..r_{u_{s+1}}]}(z) 
\leq 0
\end{eqnarray*}
by binary search, which requires $O(\log n)$ time
because the term $\theta^{\ell_{u_{s+1}}, +, [\ell_{u_{s+1}}+1..r_{u_{s+1}}]}(z)$ 
is computed in $O(1)$ time
for $z = b^{u_{s+1},+}_{p}$ by using the list of TYPE~I. 
Then, applying Lemma~\ref{lem:type1}, 
we obtain an integer $h_3$ such that 
$\theta^{\ell_{u_{s+1}}, +, [\ell_{u_{s+1}}+1..r_{u_{s+1}}]}(\delta')
=\theta^{\ell_{u_{s+1}}, +, h_3}(\delta')$ holds 
in $O(\log n)$ time. 
It implies that the value $\delta$ is the intersection point of two line segments 
$\bar{\theta}^{\ell_{u_{s+1}}, +, h_2}(C_{i, \ell_{u_{s+1}}},z)$ and 
$\theta^{\ell_{u_{s+1}}, +, h_3}(z)$, which is computed by elementally calculation.\\

\noindent
{\bf Substep 3.} Finally, we 
construct ${\cal I}^{+}_{u_{s+1}, 1}$, ${\cal I}^{+}_{u_{s+1}, 2}$ and ${\cal I}^{+}_{u_{s+1}, 3}$
from intervals ${\cal I}^{+}_{u_t}$ for $t \in [1..s]$, 
${\cal J}_1 = (\gamma, W_{j-1}]$, ${\cal J}_2 = (W_{\ell_{u_{s+1}}}, \delta]$, and ${\cal J}_3 = (\delta, W_{j-1}]$. 

Recall that we have 
\begin{eqnarray*}
\theta^{i, +, [i+1..\ell_{u_{s+1}}]}(z) = 
\left\{
\begin{array}{ll}
\theta^{i, +, [\ell_{u_t}..r_{u_t}]}(z) & \text{ if } z \in {\cal I}^{+}_{u_t}, \\
\theta^{i, +, \ell_{u_{s+1}}}(z) & \text{ if } z \in {\cal J}_1,
\end{array}
\right.
\end{eqnarray*}
and 
\begin{eqnarray*}
\theta^{i,+,[\ell_{u_{s+1}}+1..r_{u_{s+1}}]}(z)
= 
\left\{
\begin{array}{ll}
\bar{\theta}^{\ell_{u_{s+1}},+,[\ell_{u_{s+1}}+1..q-1]}(C_{i, \ell_{u_{s+1}}}, z) + \tau \cdot L_{i, \ell_{u_{s+1}}} &  \text{ if } z \in {\cal J}_2, \\
\theta^{\ell_{u_{s+1}},+,[\ell_{u_{s+1}}+1..r_{u_{s+1}}]}(z)   + \tau \cdot L_{i, \ell_{u_{s+1}}}  &  \text{ if } z \in {\cal J}_3.
\end{array}
\right.
\end{eqnarray*}

Let $\Delta(z)$ denote a function $\theta^{i, +, [i+1..\ell_{u_{s+1}}]}(z) - \theta^{i,+,[\ell_{u_{s+1}}+1..r_{u_{s+1}}]}(z)$.
We confirm that $\Delta(z)$
is non-increasing on $(W_{\ell_{u_{s+1}}}, W_{j-1}]$
by comparing the slopes of line segments of functions. 
Let a value $\mu$ denote the minimum value $z$ such that $\Delta(z) \leq 0$.
We find such $\mu$ in $O(\log^2 n)$ times as follows:

First, we check if $\mu \leq \delta$ by computing $\Delta(\delta)$.
If we have $\mu \leq \delta$, then we check if $\mu \in {\cal J}_2$, (that is, $\mu > W_{\ell_{u_{s+1}}}$) by computing $\Delta(W_{\ell_{u_{s+1}}} + \epsilon)$, 
where $\epsilon$ is a sufficiently small value. 
If we have $\mu > \delta$, then we check if $\mu = W_{j-1}$ 
by computing $\Delta(W_{j-1})$, 
Note that all the above computation for $\Delta(z)$ are done in $O(\log n)$ time by Lemmas~\ref{lem:type1} and~\ref{lem:type3}. 

Thus, we have the four cases 
for setting three intervals ${\cal I}^{+}_{u_{s+1}, 1}$, ${\cal I}^{+}_{u_{s+1}, 2}$ and ${\cal I}^{+}_{u_{s+1}, 3}$:

\begin{enumerate}[\textrm{Case} (i)]
\item If $\mu = W_{\ell_{u_{s+1}}}$, then we set 
${\cal I}^{+}_{u_{s+1}, 2} = {\cal J}_2$, ${\cal I}^{+}_{u_{s+1}, 3} = {\cal J}_3$, 
and set ${\cal I}^{+}_{u_{s+1}, 1}$ to $(\gamma, \mu]$ if $\gamma < \mu$, 
to the empty interval $(\mu, \mu]$, otherwise. 

\item If $\mu \in {\cal J}_2 = (W_{\ell_{u_{s+1}}}, \delta]$, then 
then we find the minimum value $\mu' \in \{W_h \mid W_{h} \in {\cal J}_2\}$ such that $\Delta(\mu') \leq 0$ by binary search, which takes $O(\log^2 n)$ time. 
We find an integer $H$ such that 
\begin{eqnarray*}
\theta^{i,+,[\ell_{u_{s+1}}+1..r_{u_{s+1}}]}(\mu')
&=& 
\bar{\theta}^{\ell_{u_{s+1}},+,[\ell_{u_{s+1}}+1..q-1]}(C_{i, \ell_{u_{s+1}}}, \mu') + \tau \cdot L_{i, \ell_{u_{s+1}}} \\
&=& 
\bar{\theta}^{\ell_{u_{s+1}},+,H}(C_{i, \ell_{u_{s+1}}}, \mu') + \tau \cdot L_{i, \ell_{u_{s+1}}}
\end{eqnarray*}
in $O(\log n)$ time by Lemma~\ref{lem:type3}. 
Thus, $\mu$ is the minimum value $z$ such that 
\[
\theta^{i, +, [i+1..\ell_{u_{s+1}}]}(z)
 - \bar{\theta}^{\ell_{u_{s+1}},+,H}(C_{i, \ell_{u_{s+1}}}, z) - \tau \cdot L_{i, \ell_{u_{s+1}}} \leq 0, 
\]
and can be obtained in $O(\log^2 n)$ by the same way of the first step.

Finally, we set ${\cal I}^{+}_{u_{s+1}, 1}$ to $(\gamma, \mu]$ if $\gamma < \mu$, 
to the empty interval $(\mu, \mu]$, otherwise, 
and set ${\cal I}^{+}_{u_{s+1}, 2} = (\mu, \delta]$ and ${\cal I}^{+}_{u_{s+1}, 3} = {\cal J}_3$. 

\item If we have $\mu \in {\cal J}_3 = (\delta, W_{j-1}]$, then we find the minimum value $\mu'$ over the breakpoints $\{b^{u_{s+1}, +}_{p}\}$ of $\theta^{\ell_{u_{s+1}}, +, [\ell_{u_{s+1}}+1..r_{u_{s+1}}]}(z)$
such that $\Delta(\mu') \leq 0$ holds by binary search, which requires $O(\log n)$ time. 
Then, applying Lemma~\ref{lem:type1}, 
we obtain an integer $H$ such that 
$\theta^{\ell_{u_{s+1}}, +, [\ell_{u_{s+1}}+1..r_{u_{s+1}}]}(\mu')
=\theta^{\ell_{u_{s+1}}, +, H}(\mu')$ holds in $O(\log n)$ time. 
Thus, $\mu$ is the minimum value $z$ such that 
\[
\theta^{i, +, [i+1..\ell_{u_{s+1}}]}(z) - 
\theta^{\ell_{u_{s+1}}, +, H}(z) - \tau \cdot L_{i, \ell_{u_{s+1}}} \leq 0, 
\]
and can be obtained in $O(\log^2 n)$ by the same way of the first step. 

Finally, we set ${\cal I}^{+}_{u_{s+1}, 1}$ to $(\gamma, \mu]$ if $\gamma < \mu$, 
to the empty interval $(\mu, \mu]$, otherwise, 
and ${\cal I}^{+}_{u_{s+1}, 2}$ to the empty set $(\mu, \mu]$, 
and ${\cal I}^{+}_{u_{s+1}, 3} = (\mu, W_{j-1}]$. 

\item If $\mu = W_{j-1}$, then we set ${\cal I}^{+}_{u_{s+1}, 1} = {\cal J}_1$, 
and ${\cal I}^{+}_{u_{s+1}, 2}$, ${\cal I}^{+}_{u_{s+1}, 3}$ 
to the empty interval $(W_{j-1}, W_{j-1}]$.
\end{enumerate}

For Cases~{(i)--(iii)}, we need to change ${\cal I}^{+}_{u_t, h}$ 
for $t \in [1..s]$ and $h \in [1..3]$ as follows: 
We change ${\cal I}^{+}_{u_t, h} = (\alpha_{t,h}, \beta_{t,h}]$
to $(\alpha_{t,h}, \mu]$ if $\mu \in [\alpha_{t,h}, \beta_{t,h}]$, 
and to the empty interval $(\mu, \mu]$ if $\mu < \alpha_{t,h}$.

\subsubsection{Computation Time for Phase 2}

As shown above, each inductive step requires $O(\log^2 n)$ time. 
Therefore, recalling $m=O(\log n)$, Phase~{2} requires $O(\log^3 n)$ time.

\subsection{Phase 3}\label{app:phase3}

The task of Phase~{3} is to compute the pseudo-intersection point of $\theta^{i,+}(z)$ and $\theta^{j,-}(z)$ on $[W_{i}, W_{j-1}]$. 
By Lemma~\ref{lem:convex}, $\theta^{i,+}(z)$ and $\theta^{j,-}(z)$ pseudo-intersect on $[W_{i-1},W_j]$. 
Let $z^*$ be the pseudo-intersection point of $\theta^{i,+}(z)$ and $\theta^{j,-}(z)$. 
We then see that $z^*$ is the pseudo-intersection point of $\theta^{i,+,[i+1..j]}(z)$ and $\theta^{j,-,[i..j-1]}(z)$ on $[W_{i-1},W_j]$. 

Let us suppose that in Phase 2, we have obtained 
all intervals ${\cal I}^{+}_{u,h} = (\alpha^{+}_{u,h}, \beta^{+}_{u,h}]$ that satisfy the condition~\eqref{eq:interval_condition_phase2}
for $u \in U$ and $h \in [1..3]$. 
We also have obtained all intervals ${\cal I}^{-}_{u,h} = [\alpha^{-}_{u,h}, \beta^{-}_{u,h})$ 
that satisfy the symmetric condition 
of~\eqref{eq:interval_condition_phase2}.

First, we 
find the minimum value $\beta^{+} \in \{ \beta^{+}_{u,h} \}$ such that
\begin{eqnarray}\label{eq:alg_gen_phase3.1}
\theta^{i,+,[i+1..j]}(\beta)-\theta^{j,-,[i..j-1]}(\beta) \ge 0.
\end{eqnarray}
To do this, 
we apply a binary search that takes $O(\log^2 n)$ time
by calculating  $\theta^{i,+,[i+1..j]}(\beta^{+}_{u,h})$ and $\theta^{j,-,[i..j-1]}(\beta^{+}_{u,h})$ in $O(\log n)$ time,
for $u \in U$ and $h \in [1..3]$ as follows:

For the calculation of
$\theta^{i,+,[i+1..j]}(\beta^{+}_{u,h})$, 
we consider the three cases, depending on an integer $h$: 
If $h=3$, because we have $\theta^{i,+,[i+1..j]}(\beta^{+}_{u,3})=\theta^{i,+,\ell_{u}}(\beta^{+}_{u,3})$, it can be computed in $O(1)$ time. 
If $h=2$, we have 
\begin{eqnarray*}
\theta^{i,+,[i+1..j]}(\beta^{+}_{u,2}) =
\bar{\theta}^{\ell_{u},+,[\ell_{u}+1..r_{u}]}
(C_{i, \ell_{u}},\beta^{+}_{u,2})+\tau \cdot L_{i,\ell_{u}}, 
\end{eqnarray*}
which can be computed in $O(\log n)$ time using the information of TYPE~{III} stored at $u$ by Lemma~\ref{lem:type3}. 
If $h=3$, we have 
\begin{eqnarray*}
\theta^{i,+,[i+1..j]}(\beta^{+}_{u,1})=\theta^{\ell_{u},+,[\ell_{u}+1..r_{u}]}(\beta^{+}_{u,1})+\tau \cdot L_{i,\ell_{u}},
\end{eqnarray*}
which is computed in $O(\log n)$ time using the information of TYPE~{I} stored at $u \in {\cal T}$ by Lemma~\ref{lem:type1}.
Therefore, we can calculate $\theta^{i,+,[i+1..j]}(\beta^{+}_{u,h})$ in $O(\log n)$ time.

We can also calculate $\theta^{j,-,[i..j-1]}(\beta^{+}_{u,h})$ in $O(\log n)$ time.  
First, we find an interval ${\cal I}^{-}_{u',h'}$ contains $\beta^{+}_{u,h}$ by binary search. 
We can compute $\theta^{j,-,[i..j-1]}(\beta^{+}_{u,h})$
in $O(\log n)$ time by the similar way of the above computation
according to an integer $h'$.

When the value $\beta^{+} = \beta^{+}_{u,h}$, 
it implies that $z^{*} \in {\cal I}^{+}_{u,h}$. 
Next, we find an integer $H^{+}$ such that $\theta^{i,+,[i+1..j]}(z^*)=\theta^{i,+,H^{+}}(z^*)$
in $O(\log^2 n)$ time as follows: 
If $h = 1$, that is, $z^{*} \in {\cal I}^{+}_{u,1}$, then we see $H = \ell_{u}$, directly.
If $z^* \in {\cal I}^{+}_{u, 2} = (\alpha^{+}_{u_2}, \beta^{+}_{u,2}]$, then we have 
\[
\theta^{i, +, [i..j-1]}(z) 
= \bar{\theta}^{\ell_{u}, +, [\ell_{u}+1..r_{u}]}(C_{i, \ell_{u}}, z) + \tau \cdot L_{i, \ell_{u}}.
\]
Since function $\bar{\theta}^{\ell_{u}, +, [\ell_{u}+1..r_{u}]}(C_{i, \ell_{u}}, z)$
consists of line segments with the same slope $1 / C_{i, \ell_{u}}$, 
all breakpoints of these functions on ${\cal I}^{+}_{u,2}$ are 
contained in $\{\alpha^{+}_{u,2}, \beta^{+}_{u,2}\} \cup \{W_h \mid h \in [1..n] \text{ and } W_h \in {\cal I}^{+}_{u,2}\}$.
We find the minimum value $z' \in \{\alpha_{t,2}, \beta_{t,2}\} \cup \{W_h \mid h \in [1..n] \text{ and } W_h \in {\cal I}^{+}_{u_t,2}\}$ such that 
\begin{eqnarray*}
\theta^{i,+,[i+1..j]}(z)-\theta^{j,-,[i..j-1]}(z)
= \bar{\theta}^{\ell_{u}, +, [\ell_{u}+1..r_{u}]}(C_{i, \ell_{u}}, z) + \tau \cdot L_{i, \ell_{u}} - \theta^{j, -, [i..j-1]}(z) \leq 0
\end{eqnarray*}
by binary search, which takes $O(\log^2 n)$ time because
the first term $\bar{\theta}^{\ell_{u}, +, [\ell_{u}+1..r_{u}]}(C_{i, \ell_{u}}, z)$ 
can be computed in $O(\log n)$ time for any $z$ by Lemma~\ref{lem:type3}
and the last term $\theta^{j, -, [i..j-1]}(z)$ can be computed in $O(\log n)$ time as the previous step. 
Using the value $z'$, we obtain an integer $H^{+}$ such that 
$\theta^{i,+,[i+1..j]}(z') = \bar{\theta}^{\ell_{u}, +, H^+}(C_{i, \ell_{u}}, z') + \tau \cdot L_{i, \ell_{u}}$ holds 
by Lemma~\ref{lem:type3} in $O(\log n)$ time. 
Note that we also have
$\theta^{i,+,[i+1..j]}(z^{*}) = \bar{\theta}^{\ell_{u}, +, H^+}(C_{i, \ell_{u}}, z^{*}) + \tau \cdot L_{i, \ell_{u}} = \theta^{i,+,H^{+}}(z^{*})$. 
Let us consider that $z^* \in {\cal I}^{+}_{u, 3}$.
For any $z \in {\cal I}^{+}_{u, 3}$, we have
\[
\theta^{i,+,[i+1..j]}(z) = \theta^{\ell_{u}, +, [\ell_{u}+1..r_u]}(z) + \tau \cdot L_{i, \ell_{u}}.
\]
We find adjacent breakpoints $b$ and $b'$ of $\theta^{i,+,[\ell_{u}+1..r_u]}(z)$ on ${\cal I}^{+}_{u,3}$ in $O(\log^2 n)$ time, 
such that 
\begin{eqnarray*}
\theta^{i,+,[i+1..j]}(b)-\theta^{j,-,[i..j-1]}(b) \leq 0 \text{ and }
\theta^{i,+,[i+1..j]}(b')-\theta^{j,-,[i..j-1]}(b') \geq 0
\end{eqnarray*}
hold. This means that $z^* \in [b,b']$ and then 
we can immediately obtain an integer $H^{+}$ 
such that 
$\theta^{i,+,[i+1..j]}(z^{*}) = \theta^{\ell_{u}, +, H^+}(z^{*}) + \tau \cdot L_{i, \ell_{u}} = \theta^{i,+,H^{+}}(z^{*})$, 
by using the information of TYPE~{I}.

In a symmetric manner, we find the interval 
${\cal I}^{-}_{u',h'}$ that contains $z^{*}$
and an integer $H^{-}$ such that $\theta^{j,-,[i..j-1]}(z^*)=\theta^{j,-,H^{-}}(z^*)$ 
in $O(\log^2 n)$ time. 

We then see that $z^*$ is the pseudo-intersection point of $\theta^{i,+,H^{+}}(z)$ and $\theta^{j,-,H^{-}}(z)$ on $[W_{H^{+}-1},W_{H^{-}}]$, which can be computed in $O(1)$ time.
In total, Phase~{3} requires $O(\log^2 n)$ time.

\subsection{Phase 4}\label{app:phase4}

As in Phase 3, let us suppose that in Phase 2, we have obtained 
all intervals ${\cal I}^{+}_{u,h} = (\alpha^{+}_{u,h}, \beta^{+}_{u,h}]$ 
and ${\cal I}^{-}_{u,h} = [\alpha^{-}_{u,h}, \beta^{-}_{u,h})$ 
that satisfy the condition~\eqref{eq:interval_condition_phase2} and its symmetric one 
for $u \in U$ and $h \in [1..3]$. 
In Phase~3, we also have obtained $z^*$, 
and two intervals ${\cal I}^{+}_{u^{+},h^{+}}$ and ${\cal I}^{-}_{u^{-},h^{-}}$ such that 
$z^* \in {\cal I}^{+}_{u^{+},h^{+}} \cap {\cal I}^{-}_{u^{-},h^{-}}$. 

By the definition of~\eqref{eq:at_subpath} we have 
$\Phi^{i,j}(z^*) = 
\int_0^{z^*} \theta^{i,+}(z)dz + \int_{z^*}^{W_n} \theta^{j,-}(z)dz.
$
For each integral, we have
\begin{eqnarray}\label{eq:alg_gen_phase4.1}
\int_0^{z^*} \theta^{i,+}(z)dz 
& = & 
\sum_{\substack{u \in U, h \in [1..3] s.t. \\ \beta^{+}_{u,h} \leq z^*} }
\int_{\alpha^{+}_{u,h}}^{\beta^{+}_{u,h}} \theta^{i,+}(z)dz + \int_{\alpha^+_{u^{+}, h^{+}}}^{z^*}(z)dz
\end{eqnarray}
and 
\begin{eqnarray}\label{eq:alg_gen_phase4.2}
\int_{z^*}^{W_n} \theta^{j,-}(z)dz & = & 
\int_{z^*}^{\beta^{-}_{u^{-}, h^{-}}} \theta^{j,-}(z)dz + 
\sum_{\substack{u \in U, h \in [1..3] s.t. \\ \alpha^{-}_{u,h} \geq z^*} }
\int_{\alpha^{-}_{u,h}}^{\beta^{-}_{u,h}} \theta^{j,-}(z)dz.
\end{eqnarray}

We show that an integral $\int_{\alpha^{+}_{u,h}}^{\beta^{+}_{u,h}} \theta^{i,+}(z)dz$
can be computed in $O(\log n)$ time. 
For the other integrals, we can compute them $O(\log n)$ time by the similar way. 
We can also calculate other integrals in \eqref{eq:alg_gen_phase4.1} and \eqref{eq:alg_gen_phase4.2} in $O(\log n)$ time in the same manner. 
These imply that Phase~{4} requires $O(\log^2 n)$ time, 
because the number of integrals in~\eqref{eq:alg_gen_phase4.1} 
and~\eqref{eq:alg_gen_phase4.2} is at most $2 \cdot |U| = O(\log n)$.

For the computation of 
$\int_{\alpha^{+}_{u,h}}^{\beta^{+}_{u,h}} \theta^{i,+}(z)dz$, 
let us consider the three cases:
If $h=1$, then 
condition~\eqref{eq:interval_condition_phase2} implies that 
we have 
\[
\int_{\alpha^{+}_{u,h}}^{\beta^{+}_{u,h}} \theta^{i,+}(z)dz
= \int_{\alpha^{+}_{u,h}}^{\beta^{+}_{u,h}} \theta^{i,+,\ell_{u}}(z)dz
= \int_{\alpha^{+}_{u,h}}^{\beta^{+}_{u,h}} 
\left( \frac{z-W_{\ell_{u}-1}}{C_{i,\ell_{u}}} + \tau \cdot L_{i,\ell_{u}} \right),
\]
which can be calculated in $O(1)$ time by the elementary calculation.
If $h=2$, then we have
\begin{eqnarray*}\label{eq:alg_gen_phase4.4}
\int_{\alpha^{+}_{u,h}}^{\beta^{+}_{u,h}} \theta^{i,+}(z)dz
&=& 
\int_{\alpha^{+}_{u,h}}^{\beta^{+}_{u,h}} 
\left( \bar{\theta}^{\ell_u,+,[\ell_u+1.. r_u]}(C_{i, \ell_{u}},z) + \tau \cdot L_{i,\ell_u} \right) dz \notag \\
& & \hspace{-2.9cm} = \int_{0}^{\alpha^{+}_{u,h}} \bar{\theta}^{\ell_u,+,[\ell_u+1.. r_u]}(C_{i,\ell_{u}},z)dz - \int_{0}^{\beta^{+}_{u,h}} \bar{\theta}^{\ell_u,+,[\ell_u+1.. r_u]}(C_{i,\ell_{u}},z)dz
+ \left( \beta^{+}_{u,h}-\alpha^{+}_{u,h}\right) \cdot \tau \cdot L_{i,\ell_u},
\end{eqnarray*}
which can be computed in $O(\log n)$ time using the information of TYPE~{IV} stored at $u \in {\cal T}$ by Lemma~\ref{lem:type4}. 
If $h=3$, we have
\begin{eqnarray*}\label{eq:alg_gen_phase4.5}
\int_{\alpha^{+}_{u,h}}^{\beta^{+}_{u,h}} \theta^{i,+}(z)dz
&=& 
\int_{\alpha^{+}_{u,h}}^{\beta^{+}_{u,h}} 
\left( \theta^{\ell_u,+,[\ell_u+1.. r_u]}(z) + \tau \cdot L_{i,\ell_u} \right) dz \notag \\
& & \hspace{-2.9cm} = \int_{0}^{\alpha^{+}_{u,h}} 
\theta^{\ell_u,+,[\ell_u+1.. r_u]}(z)dz - \int_{0}^{\beta^{+}_{u,h}} \theta^{\ell_u,+,[\ell_u+1.. r_u]}(z)dz
+ \left( \beta^{+}_{u,h}-\alpha^{+}_{u,h}\right) \cdot \tau \cdot L_{i,\ell_u},
\end{eqnarray*}
which can be computed in $O(\log n)$ time 
using the information of TYPEs~{I} and~{II} stored at node $u \in {\cal T}$ by Lemma~\ref{lem:type2}. 
Thus, we can compute an integral $\int_{\alpha^{+}_{u,h}}^{\beta^{+}_{u,h}} \theta^{i,+}(z)dz$ in $O(\log n)$ time. 

\subsection{Running Time of the Algorithm}
We give the analysis of the running time of our algorithm. 
Phase~{1} requires $O(\log n)$ time for finding all the maximal subpath nodes in $U$ for $P_{i+1, j-1}$ by Property~\ref{property2}. 
Phase~{2} requires $O(\log^3 n)$ time, 
and Phases~{3} and~{4} require $O(\log^2 n)$ time. 

Therefore, the bottle-neck for the running time of our algorithm
is Phase~{2} that requires $O(\log^3 n)$ time, 
which concludes the proof of Lemma~\ref{lem:general_query}.

\subsection{Modification for the Uniform Capacity}\label{app:uniform_edge_capacity}
When the capacities of ${\cal P}$ are uniform, 
we can improve the running time of Phase~{2} (that is the bottle-neck for our algorithm) from $O(\log^3 n)$ to $O(\log^2 n)$. 
The reason of the improvement is that 
we no need to construct intervals ${\cal I}_{u,2}$
since all $C_{i,j}$ are same as some $c$ for any $i, j$.
Thus, the cases that we need to use Lemma~\ref{lem:type3},
which take $O(\log^2 n)$ time, do not happen. 

\section{Conclusion}\label{sec:conclusion}
We remark here that our algorithms can be extended to  
\hide{
the following two problems in a dynamic flow path network:
(1)
The $k$-sink problem in which we minimize the squared-sum of evacuation time, 
i.e., 
we adopt $\int_{0}^{z} (\theta^{i, +}(t))^2 dt + \int_{z}^{W_n} (\theta^{j, -}(t))^2 dt$ as a cost function,
instead of $\Phi^{i,j}(z)$.
This problem corresponds to the well-known $k$-{\em means problem}.
(2)
}the minsum $k$-sink problem in a dynamic flow path network,
in which
each vertex $v_i$ has the cost $\lambda_i$ for locating a sink at $v_i$,
and we minimize 
$\at(\mathcal{P}, {\bf x}, {\bf d})+\sum_i \{\lambda_i \mid {\bf x} \mathrm{ \ consists \ of \ } v_i\}$.
Then, the same reduction works 
with link costs $w''(i,j)=w'(i,j)+\lambda_i$, which still satisfy the concave Monge property. 
This implies that our approach immediately gives algorithms of the same running time.


\bibliographystyle{splncs04}
\bibliography{ATKsink_COCOA_arXiv}

\end{document}